\documentclass{article}

\usepackage{amsmath}
\usepackage{amssymb}
\usepackage{amsthm}
\usepackage{enumerate}
\usepackage{amscd}
\usepackage{url}
\usepackage{xypic}
\usepackage[T1]{fontenc}
\usepackage[english,russian]{babel}
\usepackage[utf8]{inputenc}
\usepackage{tocloft}

\theoremstyle{plain}

\newtheorem{theorem}{Theorem}[section]
\newtheorem{lemma}[theorem]{Lemma}
\newtheorem{proposition}[theorem]{Proposition}
\newtheorem{corollary}[theorem]{Corollary}

\theoremstyle{definition}
\newtheorem{definition}[theorem]{Definition}
\newtheorem{remark}[theorem]{Remark}
\newtheorem{question}[theorem]{Question}
\newtheorem{example}[theorem]{Example}

\usepackage{url}

\newcommand{\beq}{\begin{equation}}
\newcommand{\eeq}{\end{equation}}

\DeclareMathOperator{\Div}{Div}
\DeclareMathOperator{\Cl}{Cl}

\DeclareMathOperator{\id}{id}
\DeclareMathOperator{\car}{char}

\newcommand{\PP}{{\mathbb{P}}}

\newcommand{\F}{\mathbb{F}}

\newcommand{\longto}{\longrightarrow}
\newcommand{\tens}{\otimes}
\renewcommand{\cL}{\mathcal{L}}
\newcommand{\cO}{\mathcal{O}}
\newcommand{\cI}{\mathcal{I}}
\newcommand{\cJ}{\mathcal{J}}
\newcommand{\cA}{\mathcal{A}}
\newcommand{\cB}{\mathcal{B}}
\newcommand{\cS}{\mathcal{S}}
\newcommand{\cT}{\mathcal{T}}

\begin{document}

\selectlanguage{english}

\title{Bilinear complexity of algebras and the Chudnovsky-Chudnovsky interpolation method}

%\author{Hugues Randriambololona\\
%ENST, Paris (``Telecom ParisTech'')}

\author{Hugues Randriambololona}
%\address{\'Ecole nationale sup\'erieure des t\'el\'ecommunications
%(``Telecom ParisTech'')\\
%\& LTCI CNRS UMR 5141\\
%46 rue Barrault, 75634 Paris cedex 13, France}

\maketitle

\begin{abstract}
%We give a new generalization of the Chudnovsky-Chudnovsky method that
%provides upper bounds on the bilinear complexity of multiplication in
%monogenous algebras over finite fields through interpolation on algebraic
%curves.
%Two key features of our method is that we allow asymmetric interpolation,
%as well as interpolation at arbitrary closed subschemes.
%This allows us to fix errors in, improve, and generalize, previous
%works of Shparlinski-Tsfasman-Vladut, Ballet, and Cenk-\"Ozbudak.
%
%Besides, generalities on interpolation systems, as well as an alternative
%method that solves certain constructiveness issues, are also discussed.
We give new improvements to the Chudnovsky-Chudnovsky method that
provides upper bounds on the bilinear complexity of multiplication in
extensions of finite fields through interpolation on algebraic
curves.
Our approach features three independent key ingredients:
\begin{itemize}
\item We allow asymmetry in the interpolation procedure.
This allows to prove, via the usual cardinality argument,
the existence of auxiliary divisors
needed for the bounds, up to optimal degree.
\item We give an alternative proof for the existence
of these auxiliary divisors, which is constructive, and works also
in the symmetric case, although it requires the curves to have sufficiently
many points.
\item We allow the method to deal not only with extensions of finite fields,
but more generally with monogenous algebras over finite fields. This leads
to sharper bounds, and is designed also to combine well with
base field descent arguments in case the curves do not have sufficiently
many points.
\end{itemize}
As a main application of these techniques,
we fix errors in, improve, and generalize, previous works of
Shparlinski-Tsfasman-Vladut, Ballet, and Cenk-\"Ozbudak.
Besides, generalities on interpolation systems, as well as on
symmetric and asymmetric bilinear complexity, are also discussed.
\end{abstract}

%\maketitle

%\setlength\cftparskip{-2pt}

{
\setlength\cftbeforesecskip{1pt}
\renewcommand{\cftsecfont}{\normalfont}
\tableofcontents
}

\section*{Introduction}

The bilinear complexity $\mu(\cA/K)$ of a finite-dimensional algebra $\cA$
over a field $K$ measures the essential minimal number
of two-variable multiplications in $K$ needed to
perform a multiplication in $\cA$, and considering other operations,
such as multiplication by a constant, as having no cost.
More intrinsically, it can be defined as the rank of the tensor in
\beq
\cA\tens\cA^\vee\tens\cA^\vee
\eeq
naturally deduced from the multiplication map in $\cA$.

The study of $\mu(\cA/K)$, and the effective derivation of
multiplication algorithms, are of both theoretical and practical
importance. Pioneering works in this field are Karatsuba's
algorithm~\cite{Karatsuba} for integer and polynomial multiplication,
and Strassen's algorithm~\cite{Strassen69} for matrix multiplication.

There are (at least) two ways in which these questions could be addressed
from an algebraic geometry point of view.
These two approaches are seemingly unrelated, although,
to the author's knowledge, possible links between the two have never
been seriously studied (nor will they be here).
The first one is to
consider tensors of rank $1$ as defining points of a certain Segre variety,
and tensors of higher rank, points of its successive secant varieties.
This leads to deep and beautiful problems \cite{Strassen83,Landsberg}, but we will not be interested
in this approach here.
The second one is through the theory of interpolation. Karatsuba's
algorithm may be interpreted as follows: evaluate the polynomials
at the points $0,1,\infty$ of the projective line,
multiply these values locally,
and interpolate the results to reconstruct the product polynomial.
Replacing the line with algebraic curves of higher genus
allowed Chudnovsky and Chudnovsky in~\cite{ChCh} to first prove
that the bilinear complexity of multiplication in certain extensions
of finite fields grows at most linearly with the degree.
For example, letting $\mu_q(n)=\mu(\F_{q^n}/\F_q)$, their
result implies
\beq
\label{ChCh_square_orig}
\liminf_{n\to\infty}\frac{1}{n}\mu_q(n)\leq 2\left(1+\frac{1}{\sqrt{q}-3}\right)
\eeq
for $q\geq25$ a square.

\medskip

Several improvements and variants of
the Chudnovsky-Chudnovsky algorithm were then proposed
by various authors
in order to give sharper or more general asymptotic,
as well as non-asymptotic, upper bounds. 
Roughly speaking, they all rely on the following three ingredients:
\begin{enumerate}[a)]
\item A ``generic'' interpolation process which explains how to derive
these upper bounds from the existence, postulated a priori, of certain  
geometric objects. These objects are:
\item Algebraic curves having ``good'' parameters, meaning, most of the time,
that they have sufficiently many points of various degrees, and controlled
genus.
\item Divisors on these curves, such that certain evaluation maps associated
to them are injective or surjective. Often this can be reformulated as
requiring the existence of systems of simultaneously zero-dimensional or
non-special divisors of a certain form and appropriate degree.
\end{enumerate}
These three points are important. However remark that
a well-designed algorithm in a) should make the existence
of the objects b) and c) it needs easier to check.
In this paper we will give new contributions
to a), and also to c), and then proceed to
some direct, but hopefully already significant, applications
%(although we will not try to be exhaustive then: we will
%only present the ones we consider the most direct and significative,
%hoping there will be enough such for the reader).
(further applications could be given, but they require combination
with quite different methods, so they will be treated elsewhere). 

\smallskip

Our main technical results are Theorems~\ref{ChCh+gen}
and~\ref{th=g-1} below.

\smallskip

Theorem~\ref{ChCh+gen} is our main contribution to a). There we present
a generalization of the Chudnovsky-Chudnovsky algorithm that has
two new features:
\begin{itemize}
\item We allow interpolation at arbitrary closed subschemes of the curve
in a uniform way.
The original method of Chudnovsky-Chudnovsky used only points of degree $1$,
with multiplicity $1$. Variants introduced by Ballet-Rolland and Arnaud
allowed interpolation at points of higher degree, or with higher multiplicity.
These improvements were combined and further generalized by Cenk-\"Ozbudak
in \cite{CO}. However, somehow, Cenk-\"Ozbudak still deal with degree $m$
and multiplicity $l$ separately since they use two parameters, $\mu_q(m)$
and $\widehat{M}_q(l)$, for them. Here we introduce a new quantity,
\beq
\mu_q(m,l),
\eeq
the bilinear complexity of the algebra $\F_{q^m}[t]/(t^l)$ over $\F_q$,
to deal with both at the same time. This leads ultimately to improved
bounds and is especially useful
when combined, for example, with descent arguments, such as the
ones used in~\cite{BR2004,BLR,BP}. Another indication of the naturality of our
approach is that these $\mu_q(m,l)$ can be made to appear on both
sides of our inequalities. This means, not only do we have upper bounds
\emph{in terms} of these $\mu_q(m,l)$, but at the same time we can also
derive upper bounds \emph{on} them.
\item We allow asymmetry when lifting the elements to be multiplied,
even if the multiplication law is commutative
(as is permitted by the very definition of bilinear complexity).
This has dramatic
consequences for applications since it makes 
the existence of the divisors mentioned in~c) above
much easier to prove.
Technically speaking, classical ``symmetric'' variants of the
Chudnovsy-Chudnovsky algorithm (starting from the original) suppose
given two effective divisors $G$ and $G'$ and ask for the existence
of an auxiliary divisor $D$ such that:
\beq
\label{systemesym}
\textstyle\textrm{\begin{minipage}{0.82\textwidth}
-- $\;\;D-G'$ is non-special\\
-- $\;\;2D-G$ is zero-dimensional.
\end{minipage}}
\eeq
In our asymmetric version, we ask for \emph{two} divisors $D_1,D_2$
such that:
\beq
\label{systemeasym}
\textstyle\textrm{\begin{minipage}{0.82\textwidth}
-- $\;\;D_1-G'$ and $D_2-G'$ are non-special\\
-- $\;\;D_1+D_2-G$ is zero-dimensional.
\end{minipage}}
\eeq
As explained below, this small change
allows us at once to fill a gap in the proof of bounds claimed
by Shparlinski-Tsfasman-Vladut \cite{STV} and Ballet
\cite{Ballet2008,Balletnote}.
\end{itemize}

Then Theorem~\ref{th=g-1} combines Theorem~\ref{ChCh+gen} with general
existence results for divisors as asked above, leading to bounds that
depend only on the number of points of the curve, in a somehow optimal
way. 
To be more precise, while all divisors of negative degree are zero-dimensional
(and likewise all divisors of degree more than $2g-1$ are non-special),
for the bounds on the complexity
to be as sharp as possible, one needs the divisors involved
to be of degree as near to $g-1$ as possible.

Shparlinski-Tsfasman-Vladut, and later also Ballet, claimed they
were able to solve system~\eqref{systemesym} up to degree $g-1$
(or at least, asymptotically in \cite{STV},
while exactly in \cite{Ballet2008}). For this they use a cardinality
argument. They consider the map that sends the linear equivalence
class $[D]$ to the class $[2D-G]$, and from this, deduce that
the number of linear equivalence classes of $D$ such that $2D-G$
is not zero-dimensional is not more than the number of effective
divisors of the corresponding degree. However this inference is
incorrect, because the map $[D]\mapsto[2D-G]$ is not injective.
Taking this non-injectivity into account multiplies their bound
by the $2$-torsion order of the class group, which ruins the argument.

This error was first mentioned in a preprint of
Cascudo-Cramer-Xing, although this discussion was removed from the final
version of their paper. However it can still be found in Cascudo's
PhD dissertation \cite{Cascudo}, Chap.~12.

On the other hand, our new asymmetric system~\eqref{systemeasym}
is much easier to solve. Indeed, the divisors $D_1$ and $D_2$
can then be constructed one at a time, there is no multiplication-by-$2$
map in the class group involved, and the cardinality argument works
smoothly. This allows us, under very mild assumptions,
to solve system~\eqref{systemeasym} up to
degree exactly $g-1$, which is optimal, and ultimately,
to complete the proof of the bounds claimed
in \cite{Ballet2008,Balletnote,STV}
(except for one, where there is another error, discussed in the text).
These repaired bounds now form our Corollary~\ref{corBallet}
and Theorems~\ref{STVm} and~\ref{STVM}.
For example, \eqref{ChCh_square_orig} can now be replaced safely with the
new estimate (first claimed in \cite{STV})
\beq
\label{STV_orig}
\limsup_{n\to\infty}\frac{1}{n}\mu_q(n)\leq 2\left(1+\frac{1}{\sqrt{q}-2}\right)
\eeq
for $q\geq9$ a square.

\medskip

A small drawback of this cardinality argument, already mentioned
in~\cite{STV}, is its non-constructiveness.
Also, for some applications, it might appear unsatisfactory
to get only asymmetric multiplication algorithms for an algebra
in which the multiplication law is commutative.
So we propose an alternative
method, more constructive, that solves system~\eqref{systemeasym},
as well as the original symmetric system~\eqref{systemesym},
also up to degree exactly $g-1$,
although only under more restrictive assumptions.
This alternative construction, that relies on the theory of
Weierstrass gap and order sequences,
is a straightforward adaptation of a method 
previously developed by the author in another context \cite{21sep}.
In doing so we are also led to stress the distinction between
the usual bilinear complexity, and a more restricted notion of
\emph{symmetric} bilinear complexity.
For example, our symmetric variant of \eqref{STV_orig} yields
\beq
\label{STV_sym}
\limsup_{n\to\infty}\frac{1}{n}\mu_q^{\textrm{sym}}(n)\leq 2\left(1+\frac{1}{\sqrt{q}-2}\right)
\eeq
for $q\geq49$ a square (note the stronger restriction on $q$).

\medskip

Besides these two main Theorems~\ref{ChCh+gen} and~\ref{th=g-1}
and their applications
in Corollary~\ref{corBallet}
and Theorems~\ref{STVm} and~\ref{STVM},
other topics of possible interest discussed in this paper 
include a fairly general presentation of interpolation systems
in Section~\ref{section2}, as well as a study of low degree
(or low genus) examples in Section~\ref{section4} that clarifies
and improves statements of~\cite{CO}.

\medskip

Before we finish this Introduction, we would like to mention
the very close links that exist between this domain and other
areas of mathematics and theoretical computer science.
One first such area is coding theory, and
more precisely the theory of \emph{intersecting codes}.
The link between multiplication algorithms and intersecting codes
was first stressed in \cite{BD} and \cite{LW}.
More important, in \cite{Xing2002},
Xing studied intersecting codes arising from algebraic curves,
and he gave a criterion for their existence,
that reduces essentially to the second part of
system~\eqref{systemesym}.
Hence here also the $2$-torsion in the class group is an obstruction
to get optimal parameters
(see \cite{ITW2010} for elaborations on this).
This problem was essentially solved, or more properly,
bypassed by the author in \cite{21sep}
with the method discussed above (although the analog problem
for $t$-torsion, $t\geq3$, is still open). 

Another such area is cryptography with the theory of
\emph{linear secret sharing systems with multiplication property},
in particular within the framework of
\emph{secure multi-party computation} \cite{CDM2000}.
In one direction, to optimize the parameters of these systems,
multiplication algorithms with low bilinear complexity are sometimes
required. In the other direction, secure multi-party computation
schemes based on algebraic curves were introduced by Chen and Cramer
in~\cite{ChenCramer}, and the design of these schemes also involves
a system similar to~\eqref{systemesym}.
And again, the $2$-torsion in the class group is an obstruction
to get optimal parameters \cite{Cascudo,CCX2011}.
It would be interesting to check how the tools introduced
in the present work could be put to use in this context. 
%For example one can introduce asymmetry by
%using \emph{two} divisors $D_1$ and $D_2$ on the same curve $X$,
%defining two linear secret sharing systems $L_1$ and $L_2$.
%When two elements $x_1,x_2$
%are to be multiplied, the first should be shared with $D_1$ and
%the second with $D_2$.
%If this is not already the case, then a
%\emph{resharing} step, as first introduced in \cite{CDM2000},
%is necessary. This might look like an additional cost, however,
%remark that even in classical ``symmetric'' multi-party computation
%schemes, when two elements $x_1,x_2$ shared with the same divisor $D$
%are multiplied, then their product $x_1x_2$ is shared with $2D$,
%not with $D$,
%hence a resharing step is already necessary if further multiplications
%are to be performed on it.
%Hence in most applications this additional cost will be negligible,
%while on the other hand, the analogue of \eqref{systemeasym} in
%this context being again easier than the analogue of \eqref{systemesym},
%we obtain systems of asymptotic parameters better than the ones of
%\cite{CCX2011}.

\medskip

\emph{Conventions.}
In this text we make free use of the language of modern algebraic
geometry: schemes, sheaves, and cohomology.
Admittedly, the only place where this is necessary is at 
the end of Section~\ref{section2}, while designing interpolation
systems from higher dimensional algebraic varieties, and this
point is quite secondary in our presentation.
From Section~\ref{section3} on, we deal only with curves,
and everything could be equally well expressed in the language
of function fields in one indeterminate.
We made the choice to stick to the geometric point of view,
but, keeping in mind that application oriented readers might be
more familiar with the function field terminology, we tried to
keep the level of exposition accessible so that translation
from one language to the other would remain easy.
As standard references for these subjects we advise \cite{Hartshorne}
for the general geometric language and \cite{Stichtenoth} for the
function field approach in the case of curves.

\section{Tensor rank and bilinear complexity}

\begin{definition}
Let $K$ be a field, and $E_0,\dots,E_s$ be finite-dimensional
$K$-vector spaces. A non-zero element $t\in E_0\tens\cdots\tens E_s$
is said to be an \emph{elementary tensor},
or a \emph{tensor of rank~$1$},
if it can be written in the form $t=e_0\tens\cdots\tens e_s$
for some $e_i\in E_i$. More generally, the \emph{rank} of an
arbitrary $t\in E_0\tens\cdots\tens E_s$ is defined as the minimal
length of a decomposition of $t$ as a sum of elementary tensors.
\end{definition}

%(Remark that this notion of rank does not arise from the whole
%tensor space $E_0\tens\cdots\tens E_s$ alone, it really depends
%on the individual spaces $E_0,\dots,E_s$.) 

\begin{definition}
If
\beq
\alpha:E_1\times\cdots\times E_s\longto E_0
\eeq
is an $s$-linear map,
the $s$-linear complexity of $\alpha$ is defined as the tensor rank
of the element
\beq
\widetilde{\alpha}\in E_0\tens E_1^\vee\tens\cdots\tens E_s^\vee
\eeq
naturally deduced from $\alpha$.
\end{definition}

For $s=1$, these notions are very well understood
(they reduce essentially to the rank of a matrix).
However, starting from $s=2$, they can be surprisingly difficult
to handle.

\begin{definition}
Let $\cA$ be a finite-dimensional $K$-algebra. We denote by
\beq
\mu(\cA/K)
\eeq
the bilinear complexity of the multiplication map
\beq
m_{\cA}:\cA\times\cA\longto\cA
\eeq
considered as a $K$-bilinear map.
\end{definition}

More concretely, $\mu(\cA/K)$ is the smallest integer $n$
such that there exist linear forms
$\phi_1,\dots,\phi_n$ and $\psi_1,\dots,\psi_n:\cA\longto K$,
and elements $w_1,\dots,w_n\in\cA$, such that for all
$x,y\in \cA$ one has
\beq
\label{algomult}
xy=\phi_1(x)\psi_1(y)w_1\,+\,\cdots\,+\,\phi_n(x)\psi_n(y)w_n.
\eeq
Indeed, such an expression is the same thing as a decomposition
\beq
\widetilde{m}_{\cA}=\sum_{i=1}^nw_i\tens\phi_i\tens\psi_i\:\in\:\cA\tens\cA^\vee\tens\cA^\vee
\eeq
for the multiplication tensor of $\cA$.

Remark that here, the notion of algebra is taken in its broadest sense.
However, in Proposition~\ref{schema_arbitraire},
and then from Section~\ref{section3} on,
we will only consider algebras that are associative, commutative, and
with unity.

\begin{definition}
We call \emph{multiplication algorithm}
of length $n$ for $\cA/K$ 
a collection of $\phi_i$, $\psi_i$, $w_i$ that satisfy~\eqref{algomult}.
Such an algorithm is said \emph{symmetric}
if $\phi_i=\psi_i$ for all $i$ (this can happen only if $\cA$ is commutative).
\end{definition}

The study of $\mu(\cA/K)$, and the effective derivation of
multiplication algorithms, are of both theoretical and practical
importance. Pioneering works in this field are Karatsuba's
algorithm~\cite{Karatsuba} for integer and polynomial multiplication,
and Strassen's algorithm~\cite{Strassen69} for matrix multiplication.

In practical terms, focusing on the bilinear complexity
of the multiplication in $\cA$ means according importance only to
the number of two-variable multiplications in $K$ needed to
perform a multiplication in $\cA$, and considering other operations,
such as multiplication by a constant, as having no cost.
This is a reasonable assumption although its relevance clearly depends
on the computation model.

When $\cA$ is commutative, it is sometimes convenient to favour
the study of symmetric multiplication algorithms. Thus, as $\mu(\cA/K)$
is defined as the minimal length of a 
(possibly asymmetric) multiplication algorithm for $\cA/K$,
we also introduce the following:

\begin{definition}
If $\cA$ is a finite-dimensional commutative $K$-algebra,
we define its \emph{symmetric bilinear complexity}
\beq
\mu^{\textrm{sym}}(\cA/K)
\eeq
as the minimal length of a 
symmetric multiplication algorithm for $\cA/K$.

Equivalently, it is the minimal length of a decomposition
of the multiplication tensor $\widetilde{m}_{\cA}$ as a sum
of symmetric elementary tensors, that is, of tensors
of the form $w\tens\phi\tens\phi\in\cA\tens\cA^\vee\tens\cA^\vee$.  
\end{definition}

Here we gather a few elementary properties of these
notions.
Lemma~\ref{sym-asym} shows that symmetric bilinear complexity
is well defined, and compares it with its non-symmetric counterpart.
Lemma~\ref{basic_lower_bounds} gives basic lower bounds
for $\mu(\cA/K)$, and Lemma~\ref{functorial_inequalities}
deals with some functorial properties.
Certainly most things here are already classical and can be found
from other sources.
The reader is especially refered to the foundational work \cite{Strassen73}
(and to the additional material in \cite{BD,FZ,LW,Winograd}),
or to textbooks such as \cite{BCS,deGr}, for historical details and
further results of this type.

\begin{lemma}
\label{sym-asym}
Let $\cA$ be a finite-dimensional commutative $K$-algebra.
Then $\cA$ admits a symmetric multiplication algorithm,
hence $\mu^{\textrm{sym}}(\cA/K)<\infty$ is well defined.
More precisely, it satisfies
\beq
\label{ineq_musym_1}
\mu^{\textrm{sym}}(\cA/K)\leq\frac{d(d+1)}{2}
\eeq
where $d=\dim\cA$.
If $\car K\neq 2$, then also
\beq
\label{ineq_musym_2}
\mu^{\textrm{sym}}(\cA/K)\leq 2\mu(\cA/K).
\eeq
In the other direction, we always have
\beq
\label{ineq_musym_3}
\mu(\cA/K)\leq\mu^{\textrm{sym}}(\cA/K).
\eeq
\end{lemma}
\begin{proof}
Let $e_1,\dots,e_d$ be a basis of $\cA$, and let $e_1^\vee,\dots,e_d^\vee$
be the dual basis. First remark that the multiplication tensor of $\cA$
can always be decomposed as
$\widetilde{m}_{\cA}=\sum_{i,j}(e_ie_j)\tens e_i^\vee\tens e_j^\vee$,
and since $\cA$ is commutative this can be rearranged as:
\beq
\widetilde{m}_{\cA}=\sum_{1\leq i\leq d}(e_i^{\,2})\tens e_i^\vee\tens e_i^\vee+\sum_{1\leq i<j\leq d}(e_ie_j)\tens (e_i^\vee\tens e_j^\vee+e_j^\vee\tens e_i^\vee).
\eeq
The first sum is already composed of symmetric tensors,
and the second sum can also be put in such a form since
\beq
\label{3matrices}
e_i^\vee\tens e_j^\vee+e_j^\vee\tens e_i^\vee=(e_i^\vee+e_j^\vee)\tens(e_i^\vee+e_j^\vee)-e_i^\vee\tens e_i^\vee-e_j^\vee\tens e_j^\vee.
\eeq
We plug this into the previous equality and then regroup the similar terms
to find:
\beq
\widetilde{m}_{\cA}=\sum_{1\leq i\leq d}(2e_i^{\,2}-e_is)\tens e_i^\vee\tens e_i^\vee+\sum_{1\leq i<j\leq d}(e_ie_j)\tens(e_i^\vee+e_j^\vee)\tens(e_i^\vee+e_j^\vee)
\eeq
where $s=\sum_{j=1}^ne_j$.
This gives \eqref{ineq_musym_1}.

\smallskip

Now suppose $\car K\neq 2$,
and let $w_i,\phi_i,\psi_i$ define a
multiplication algorithm of length $n=\mu(\cA/K)$ for $\cA$.
We can then write
\beq
\begin{split}
\widetilde{m}_{\cA}&=\sum_{i=1}^nw_i\tens\phi_i\tens\psi_i=\sum_{i=1}^nw_i\tens\psi_i\tens\phi_i\\
&=\frac{1}{2}\sum_{i=1}^nw_i\tens(\phi_i\tens\psi_i+\psi_i\tens\phi_i)\\
&=\frac{1}{4}\sum_{i=1}^nw_i\tens(\phi_i+\psi_i)\tens(\phi_i+\psi_i)\,-\,w_i\tens(\phi_i-\psi_i)\tens(\phi_i-\psi_i),
\end{split}
\eeq
hence \eqref{ineq_musym_2}.

\smallskip

Last, \eqref{ineq_musym_3} is trivial.
\end{proof}

\begin{remark}
\label{ex_mu<musym}
Let $K=\F_2$.
We can interpret \eqref{3matrices} as giving a decomposition
of the rank~\emph{two} symmetric matrix
$\left(\begin{array}{cc} 0 & 1 \\ 1 & 0\end{array}\right)$
as a sum of \emph{three} rank~$1$ symmetric matrices:
\beq
\left(\begin{array}{cc} 0 & 1 \\ 1 & 0\end{array}\right)
=
\left(\begin{array}{cc} 1 & 1 \\ 1 & 1\end{array}\right)
+
\left(\begin{array}{cc} 1 & 0 \\ 0 & 0\end{array}\right)
+
\left(\begin{array}{cc} 0 & 0 \\ 0 & 1\end{array}\right).
\eeq
For $K=\F_2$ it is easily seen that this decomposition
is \emph{minimal}.

As a consequence, if $\cA$ is the $2$-dimensional
commutative (but non-associative and without unity)
$\F_2$-algebra with basis $e_1,e_2$ and multiplication
defined by $e_1e_2=e_2e_1=e_1$ and $e_1^{\,2}=e_2^{\,2}=0$,
then
\beq
\mu(\cA/K)=2\:<\:\mu^{\textrm{sym}}(\cA/K)=3.
\eeq
This gives an example of strict inequality in \eqref{ineq_musym_3}.
\end{remark}

\begin{definition}
Given a multiplication algorithm as in~\eqref{algomult}, 
one associates to it two linear codes
$C_\phi$ and $C_\psi\subset K^n$, namely the images of
the evaluation maps
\beq
\begin{array}{cccc}
\phi: & \cA & \longto & \!\!\!K^n\\
& x & \mapsto & \!\!\!(\phi_1(x),\dots,\phi_n(x))
\end{array}
\quad\textrm{and}\quad
\begin{array}{cccc}
\psi: & \cA & \longto & \!\!\!K^n\\
& y & \mapsto & \!\!\!(\psi_1(y),\dots,\psi_n(y))
\end{array}
\eeq
respectively.
\end{definition}

\begin{lemma}
\label{basic_lower_bounds}
Let $\cA$ be a finite-dimensional $K$-algebra.
\begin{enumerate}[a)]
\item If $\cA$ admits a unit element,
\beq
\mu(\cA/K)\geq \dim_K\cA.
\eeq
\item If $\cA$ has no zero-divisor,
\beq
\label{Singleton}
\mu(\cA/K)\geq 2\dim_K\cA-1.
\eeq
\end{enumerate}
\end{lemma}
\begin{proof}
Consider a multiplication algorithm as in~\eqref{algomult}.
If $\cA$ admits a unit element, then $w_1,\dots,w_n$ span $\cA$,
hence the first inequality.
For the second inequality, remark that if $\cA$ has no zero-divisor,
then:
\begin{itemize}
\item the maps $\phi$ and $\psi$ must be injective,
hence the codes $C_\phi$ and $C_\psi$ have dimension
$k=\dim_K\cA$,
\item these two codes must be \emph{mutually intersecting},
that is, any non-zero $c\in C_\phi$ and $c'\in C_\psi$ must have
non-disjoint supports.
\end{itemize}
%From this second condition, we see that
%the projection of $C_\psi$ on the support of a codeword $c\in C_\phi$
%of minimal weight $w(c)=d_{\min}(C_\phi)$ must be injective,
%hence $d_{\min}(C_\phi)\geq k$, and Singleton's bound allows
%to conclude $n\geq 2k-1$ as claimed. 
By the first point, if $k>\lceil n/2\rceil$, one could find
a non-zero $c\in C_\phi$ vanishing on the first $\lceil n/2\rceil$
coordinates, and a non-zero $c'\in C_\psi$ vanishing on the last
$\lceil n/2\rceil$. These $c,c'$ would then contradict the second point.
Hence $k\leq\lceil n/2\rceil$,
which gives precisely~\eqref{Singleton}.
\end{proof}

The link between multiplication algorithms and intersecting codes
was first stressed in \cite{BD} and \cite{LW}. For more on this last topic,
see for example %\cite{CL}
\cite{21sep} and the references therein.
Another coding-theoretical view on some bilinear complexity problems
has also been proposed, through the notion of \emph{supercode},
in \cite{STV}.

\begin{lemma}
\label{functorial_inequalities}
\begin{enumerate}[a)]
\item
If $\cA$ is a finite-dimensional $K$-algebra
and $L$ an extension field of $K$,
and if we let $\cA_L=\cA\tens_K L$ considered as an $L$-algebra, then
\beq
\label{funct1}
\mu(\cA_L/L)\leq\mu(\cA/K).
\eeq
\item
If $\cA$ is a finite-dimensional $L$-algebra,
where $L$ is an extension field of $K$,
then $\cA$ can also be considered as a $K$-algebra, and
\beq
\label{funct2}
\mu(\cA/K)\leq\mu(\cA/L)\mu(L/K).
\eeq
\item
If $\cA$ and $\cB$ are two finite-dimensional $K$-algebras,
\beq
\label{funct3}
\mu(\cA\times\cB/K)\leq\mu(\cA/K)+\mu(\cB/K).
\eeq
\item
If $\cA$ and $\cB$ are two finite-dimensional $K$-algebras,
\beq
\label{funct4}
\mu(\cA\tens_K\cB/K)\leq\mu(\cA/K)\mu(\cB/K).
\eeq
\end{enumerate}
Moreover, when the algebras are commutative, then
\eqref{funct1}\eqref{funct2}\eqref{funct3}\eqref{funct4}
also hold with $\mu^{\textrm{sym}}$ in place of $\mu$.
\end{lemma}
\begin{proof}
To prove a), remark that if linear forms
$\phi_1,\dots,\phi_n$ and $\psi_1,\dots,\psi_n:\cA\longto K$
and elements $w_1,\dots,w_n\in \cA$ define
a multiplication algorithm for $\cA/K$, then the $\phi_i$ and $\psi_i$
lift to linear forms $\cA_L\longto L$, and the $w_i$ can be seen
as elements of $\cA_L$, and as such they define a multiplication
algorithm for $\cA_L/L$
of the same length~$n$.

\smallskip

To prove b) we use an analogue of the concatenation procedure
in coding theory.
Formally, suppose we are given:
\begin{itemize}
\item a multiplication algorithm of length $m$ for $L/K$, defined
by linear forms $\alpha_1\dots,\alpha_m$
and $\beta_1\dots,\beta_m:L\longto K$
and elements $l_1,\dots,l_m\in L$,
\item a multiplication algorithm of length $n$ for $\cA/L$, defined
by linear forms $\lambda_1\dots,\lambda_n$
and $\rho_1\dots,\rho_n:\cA\longto L$
and elements $a_1,\dots,a_n\in \cA$.
\end{itemize}
Then, letting $N=mn$, the two collections of $N$ linear forms
$\phi_{i,j}=\alpha_i\circ\lambda_j$
and $\psi_{i,j}=\beta_i\circ\rho_j:\cA\longto K$,
and the $N$ elements $w_{i,j}=l_ia_j\in\cA$,
for $1\leq i\leq m$ and $1\leq j\leq n$,
define a multiplication algorithm of length $N$ for $\cA/K$.
Indeed, for all $x,y\in\cA$,
\beq
xy=\sum_{1\leq j\leq n}\lambda_j(x)\rho_j(y)a_j=\sum_{1\leq j\leq n}\left(\sum_{1\leq i\leq m}\alpha_i(\lambda_j(x))\beta_i(\rho_j(y))l_i\right)a_j.
\eeq
To make the connection with concatenation in coding theory clearer,
remark that $C_\phi$ is then the concatenated code $C_\alpha\circ C_\lambda$,
and likewise $C_\psi=C_\beta\circ C_\rho$.

\smallskip

The proof of c) proceeds analogously using the notion
of direct sum of multiplication algorithms. Suppose we are given:
\begin{itemize}
\item a multiplication algorithm of length $m$ for $\cA/K$, defined
by linear forms $\phi_1\dots,\phi_m$
and $\psi_1\dots,\psi_m:\cA\longto K$
and elements $a_1,\dots,a_m\in\cA$,
\item a multiplication algorithm of length $n$ for $\cB/K$, defined
by linear forms $\lambda_1\dots,\lambda_n$
and $\rho_1\dots,\rho_n:\cB\longto K$
and elements $b_1,\dots,b_n\in \cB$.
\end{itemize}
Identify $\cA$ with the subspace $\cA\times\{0\}$
and $\cB$ with the subspace $\{0\}\times\cB$ in $\cA\times\cB$.
Then for any $x=(r,s)$ and $y=(u,v)$ in $\cA\times\cB$ we have
\beq
xy=ru+sv=\sum_{1\leq i\leq m}\phi_i(r)\psi_i(u)a_i+\sum_{1\leq j\leq n}\lambda_j(s)\rho_j(v)b_j
\eeq
hence this defines a multiplication algorithm
of length $m+n$ for $\cA\times\cB$.

\smallskip

For d) we skip the details since everything works the same:
suppose given $\phi_i,\psi_i,a_i$ and $\lambda_j,\rho_j,b_j$
as in the proof of $c)$,
then the $\phi_i\tens\lambda_j,\psi_i\tens\rho_j,a_i\tens b_j$
give a multiplication algorithm of length $N=mn$ for $\cA\tens\cB$.

\smallskip

For the last assertion, remark that if we start with symmetric
algorithms, then the constructions given above lead also
to symmetric algorithms. 
\end{proof}

%\begin{remark}
%\label{rem_constr_sym_1}
%It should be noted that the proof of Lemma~\ref{functorial_inequalities}
%gives a little more than what is stated. Namely, provided multiplication
%algorithms for the algebras appearing on the right side of the inequalities,
%it effectively
%constructs a multiplication algorithm of the ``designed'' length for the
%algebras on the left side.

%Moreover, if we start with symmetric
%algorithms, this construction also produces a symmetric algorithm.
%Thus, there is a natural notion of \emph{symmetric bilinear complexity}
%for commutative algebras
%(define $\mu^{\textrm{sym}}(\cA/K)$ as the minimal length of a symmetric
%multiplication algorithm for $\cA/K$), and the proof given here shows
%that it satisfies
%the same functorial inequalities as in the lemma.
%\end{remark}

\begin{question}
It would be interesting to have criteria for equality
in this Lemma~\ref{functorial_inequalities}.
For the inequalities in parts a) and b) (and hence also for part d),
there are non-trivial examples in which equality holds,
and others in which the inequality is strict 
(see below, or \cite{Winograd}).
A general rule does not seem obvious.
Turning to c), the author does not know any example were the inequality
is strict. In fact, the now folklore \emph{direct sum conjecture}
(see \cite{FZ,Strassen73,Winograd})
suggests there should always be equality:
\beq
\label{direct_sum_conj}
\mu(\cA\times\cB/K)\overset{?}{=}\mu(\cA/K)+\mu(\cB/K).
\eeq
Proofs are known only for some very specific classes of algebras.
The general case is still open.
\end{question}

\begin{remark}
We would like to indicate a few possible generalizations of the notions
developed so forth.

First, we worked over a field, but it is also possible to work
over a ring, or even over a more general base.
This could be of interest, for instance, if one is given a family
of tensors that vary with some parameters, and one requests
elementary decompositions for them that vary accordingly.

In another direction, one could also extend the notion of symmetry.
Given a group $G$ acting on some tensor space,
we can ask whether every $G$-invariant tensor admits a decomposition
as a sum of $G$-invariant elementary tensors (and if so,
what is the minimal length of such a decomposition).
For $G=\mathfrak{S}_2$ the symmetric group of order~$2$ acting on $\cA\tens\cA^\vee\tens\cA^\vee$
by permuting the last two factors, we saw in Lemma~\ref{sym-asym}
that this is true (although the minimal symmetric decomposition
might be longer than the non-symmetric one).
However for more general group actions this is not always possible.
The elegant counterexample that follows
is due to Cascudo~\cite{Casc_pers}:

Consider the trilinear map
\beq
\begin{array}{ccc}
\F_4\times\F_4\times\F_4 & \longto & \F_4 \\
(x,y,z) & \mapsto & xyz
\end{array}
\eeq
over $\F_2$. It defines a tensor
in $\F_4\tens\F_4^\vee\tens\F_4^\vee\tens\F_4^\vee$,
and since $\F_4$ is commutative,
this tensor is $\mathfrak{S}_3$-invariant, where $\mathfrak{S}_3$ acts
by permuting the last three factors. Suppose this tensor
admits an $\mathfrak{S}_3$-invariant elementary decomposition.
This means one can find elements $w_1,\dots,w_n\in\F_4$,
and linear forms $\phi_1,\dots,\phi_n:\F_4\to\F_2$, such that for
all $x,y,z\in\F_4$,
one has $xyz=\sum_{i=1}^n\phi_i(x)\phi_i(y)\phi_i(z)w_i$.
But then for all $x,y\in\F_4$ one finds
\beq
\begin{split}
&x^2y=\sum_{i=1}^n\phi_i(x)^2\phi_i(y)w_i \\
&xy^2=\sum_{i=1}^n\phi_i(x)\phi_i(y)^2w_i 
\end{split}
\eeq
and the two quantities on the right are equal because
all $\alpha\in\F_2$ satisfy $\alpha^2=\alpha$.
This is a contradiction since there are $x,y\in\F_4$ with $x^2y\neq xy^2$.

\end{remark}

\section{Interpolation systems}

\label{section2}

If $\cB$ is a $K$-algebra and if $E_1,E_2\subset\cB$
are two linear subspaces, we denote by $E_1E_2$ the \emph{linear span}
of the products $e_1e_2$ in $\cB$, for $e_1\in E_1$ and $e_2\in E_2$.

\begin{definition}
\label{def_interpol}
Let $\cA$ and $\cA'$ be two finite-dimensional $K$-algebras.
By an \emph{interpolation system} for $\cA'$ by $\cA$ we mean
the following data:
\begin{itemize}
\item a $K$-algebra $\cB$ (of possibly infinite dimension) equipped with
two $K$-algebra morphisms $f:\cB\longto\cA$
and $f':\cB\longto\cA'$
\item two linear subspaces $E_1,E_2\subset\cB$
\end{itemize}
satisfying the following conditions:
\begin{itemize}
\item[\emph{(i)}] the restriction $f|_{E_1E_2}:E_1E_2\longto\cA$ is injective
\item[\emph{(ii)}] the restrictions $f'|_{E_1}:E_1\longto\cA'$
and $f'|_{E_2}:E_2\longto\cA'$ are surjective.
\end{itemize}
This can be summarized with the following diagram:
\begin{equation*}
\xymatrix{
\;E_1E_2 \ar@{^{(}->}[d] &  \cB \ar[ld] \ar[rd] &  E_1,E_2 \ar@{->>}[d]\\
\cA & & \cA'
}
\end{equation*}
Such an interpolation system is said \emph{symmetric} if $E_1=E_2$.
\end{definition}

\begin{proposition}
\label{prop_interp_general}
Let $\cA$ and $\cA'$ be two finite-dimensional $K$-algebras.
Suppose there exists an interpolation system for $\cA'$ by $\cA$.
Then
\beq
\mu(\cA'/K)\leq\mu(\cA/K).
\eeq
Moreover, if $\cA$ and $\cA'$ are commutative and the interpolation
system is symmetric,
then also $\mu^{\textrm{sym}}(\cA'/K)\leq\mu^{\textrm{sym}}(\cA/K)$.
\end{proposition}
\begin{proof}
Let $\phi_1,\dots,\phi_n$,
$\psi_1,\dots,\psi_n:\cA\longto K$,
and $w_1,\dots,w_n\in\cA$ define a multiplication algorithm
for $\cA/K$, where $n=\mu(\cA/K)$.

Suppose we are given an interpolation system for $\cA'$ by $\cA$.
Thanks to properties \emph{(i)} and \emph{(ii)} above, we can choose:
\begin{itemize}
\item a retraction $\rho:\cA\longto E_1E_2$ of $f|_{E_1E_2}$
\item sections $\sigma_1:\cA'\longto E_1$ of $f'|_{E_1}$
and $\sigma_2:\cA'\longto E_2$ of $f'|_{E_2}$.
\end{itemize}
Then, for $1\leq i\leq n$, we let:
\begin{itemize}
\item $\phi'_i=\phi_i\circ f|_{E_1}\circ \sigma_1:\cA'\longto K$
\item $\psi'_i=\psi_i\circ f|_{E_2}\circ \sigma_2:\cA'\longto K$
\item $w'_i=f'(\rho(w_i))\in K$.
\end{itemize}
Then $\phi'_1,\dots,\phi'_n$,
$\psi'_1,\dots,\psi'_n$,
and $w'_1,\dots,w'_n\in\cA$ define a multiplication algorithm
for $\cA'/K$. Indeed, for any $x',y'\in\cA'$,
if we let $x=f(\sigma_1(x'))$ and $y=f(\sigma_2(y'))$, then:
\beq
\begin{split}
\sum\phi'_i(x')\psi'_i(y')w'_i&=\sum\phi_i(x)\psi_i(y)f'(\rho(w_i))\\
&=f'(\rho(\sum\phi_i(x)\psi_i(y)w_i))\\
&=f'(\rho(xy))\\
&=f'(\rho(f(\sigma_1(x'))f(\sigma_2(y'))))\\
&=f'(\rho(f(\sigma_1(x')\sigma_2(y'))))\\
&=f'(\sigma_1(x')\sigma_2(y'))\\
&=f'(\sigma_1(x'))f'(\sigma_2(y'))\\
&=x'y'.
\end{split}
\eeq
Thus $\mu(\cA'/K)\leq n$, as claimed.

For the last assertion, supposing $E_1=E_2$,
remark that if we start with a symmetric algorithm for $\cA/K$
and if we choose $\sigma_1=\sigma_2$, then the construction
gives a symmetric algorithm for $\cA'/K$.
\end{proof}

%\begin{remark}
%\label{rem_constr_sym_2}
%As in Remark~\ref{rem_constr_sym_1}, it should be noted that the proof of
%Proposition~\ref{prop_interp_general} is constructive, in that,
%given an interpolation system and a multiplication algorithm for $\cA/K$,
%it effectively constructs a multiplication algorithm for $\cA'/K$
%of the same length.

%Moreover if we start from a symmetric algorithm
%for $\cA/K$, and if the interpolation
%scheme also is symmetric ($E_1=E_2$), then, provided we choose
%$\sigma_1=\sigma_2$ in the construction, we get a symmetric
%algorithm for $\cA'/K$. Said otherwise,
%we have the following variant of the proposition:
%if $\cA$ and $\cA'$ are
%commutative and if there exists a \emph{symmetric}
%interpolation system for $\cA'$ by $\cA$,
%then
%\beq
%\mu^{\textrm{sym}}(\cA'/K)\leq\mu^{\textrm{sym}}(\cA/K).
%\eeq
%\end{remark}

\begin{corollary}
If $\cA$ is a finite-dimensional $K$-algebra,
and if $\cA'$ is a subalgebra of $\cA$, or a quotient algebra of $\cA$,
then
\beq
\mu(\cA'/K)\leq\mu(\cA/K).
\eeq
If $\cA$ is commutative,
then also $\mu^{\textrm{sym}}(\cA'/K)\leq\mu^{\textrm{sym}}(\cA/K)$.
\end{corollary}
\begin{proof}
If $\cA'$ is a subalgebra of $\cA$, define an interpolation system
by taking $E_1=E_2=\cB=\cA'$, $f$ the natural inclusion, and $f'=\id_{\cA'}$.

If $\cA'$ is a quotient algebra of $\cA$,
take $E_1=E_2=\cB=\cA$, $f=\id_{\cA}$, and $f'$ the
natural projection.
\end{proof}

The preceding corollary makes a rather trivial use of the notion of
interpolation system. We will see more interesting examples,
arising from algebraic geometry
(for which we refer to standard textbooks such as \cite{Hartshorne}),
as follows.

\begin{proposition}
\label{schema_arbitraire}
Let $X$ be an algebraic variety, or more generally
an arbitrary scheme over $K$,
and let $\Sigma$ and $\Sigma'$ be two closed subschemes
of $X$ that are finite over $K$.
Suppose there are two invertible sheaves
$\cL_1$ and $\cL_2$ on $X$ such that:
\begin{enumerate}[(i)]
\item the natural restriction map
\beq
\Gamma(X,\cL_1\tens\cL_2)\longto\Gamma(\Sigma,\cL_1\tens\cL_2)
\eeq
is injective
\item
the natural restriction maps
\beq
\Gamma(X,\cL_1)\longto\Gamma(\Sigma',\cL_1)\qquad\Gamma(X,\cL_2)\longto\Gamma(\Sigma',\cL_2)
\eeq
are surjective.
\end{enumerate}
%Write $\Sigma=\Spec\cA$ and $\Sigma'=\Spec\cA'$, where
Consider the rings
$\cA=\Gamma(\Sigma,\cO_\Sigma)$ and $\cA'=\Gamma(\Sigma',\cO_{\Sigma'})$.
Then
\beq
\label{inegalite_schema_arbitraire}
\mu(\cA'/K)\leq\mu(\cA/K).
\eeq
Moreover, if $\cL_1=\cL_2$,
then also $\mu^{\textrm{sym}}(\cA'/K)\leq\mu^{\textrm{sym}}(\cA/K)$.

A sufficient criterion for the conditions (i) and (ii) above to hold,
hence also for the conclusion~\eqref{inegalite_schema_arbitraire},
can be expressed in terms of vanishing of certain cohomology groups
as follows:
\begin{enumerate}[(i')]
\item $h^0(X,\cI(\cL_1\tens\cL_2))=0$
\item $h^1(X,\cI'\cL_1)=h^1(X,\cI'\cL_2)=0$
\end{enumerate}
where $\cI$ and $\cI'$ are the sheaves of ideals on $X$
defining $\Sigma$ and $\Sigma'$, respectively.
In fact,
(i) and (i') are equivalent, while (ii') only implies (ii)
a priori.
\end{proposition}
\begin{proof}

Remark first that $\Sigma$ and $\Sigma'$ are finite over $K$,
hence affine, and the rings $\cA$ and $\cA'$ are Artinian,
and as such they can be written as a finite direct product of local rings.
Thus any invertible module over $\cA$ or $\cA'$,
or equivalently any invertible sheaf over $\Sigma$ or $\Sigma'$, is free.
In particular, we can choose trivializations
\beq
\label{basic_triv}
\Gamma(\Sigma,\cL_1)\simeq\Gamma(\Sigma,\cL_2)\simeq\cA\qquad\Gamma(\Sigma',\cL_1)\simeq\Gamma(\Sigma',\cL_2)\simeq\cA'
\eeq
and from these, deduce, for any integers $i_1,i_2$, trivializations
\beq
\label{comp_triv}
\Gamma(\Sigma,\cL_1^{\tens i_1}\tens\cL_2^{\tens i_2})=\Gamma(\Sigma,\cL_1)^{\tens i_1}\tens\Gamma(\Sigma,\cL_2)^{\tens i_2}\simeq\cA
\eeq
\beq
\label{comp_triv'}
\Gamma(\Sigma',\cL_1^{\tens i_1}\tens\cL_2^{\tens i_2})=\Gamma(\Sigma',\cL_1)^{\tens i_1}\tens\Gamma(\Sigma',\cL_2)^{\tens i_2}\simeq\cA'.
\eeq
Consider now the bigraded algebra
\beq
\cB=\bigoplus_{i_1,i_2\geq0}\Gamma(X,\cL_1^{\tens i_1}\tens\cL_2^{\tens i_2}).
\eeq
It comes equipped with two morphisms of bigraded algebras
\beq
\cB\;\longto\bigoplus_{i_1,i_2\geq0}\Gamma(\Sigma,\cL_1^{\tens i_1}\tens\cL_2^{\tens i_2}) 
\qquad
\cB\;\longto\bigoplus_{i_1,i_2\geq0}\Gamma(\Sigma',\cL_1^{\tens i_1}\tens\cL_2^{\tens i_2}) 
\eeq
defined by the natural restriction maps, and composing with \eqref{comp_triv}
and \eqref{comp_triv'}, and then taking the sum, we get
\beq
f:\cB\longto\cA\qquad\qquad f':\cB\longto\cA'.
\eeq
Since \eqref{comp_triv} and \eqref{comp_triv'}
were defined in a compatible way from \eqref{basic_triv} as $i_1,i_2$ vary,
we see that $f$ and $f'$ are not merely morphisms of vector spaces,
they are in fact morphisms of algebras. Now we take
\beq
E_1=\cB_{1,0}=\Gamma(X,\cL_1)\quad\quad\quad E_2=\cB_{0,1}=\Gamma(X,\cL_2)
\eeq
so
\beq
E_1E_2\subset\cB_{1,1}=\Gamma(X,\cL_1\tens\cL_2)
\eeq
and conditions \emph{(i)} and \emph{(ii)} in our hypotheses imply
conditions \emph{(i)} and \emph{(ii)} in the definition of interpolation systems.
We can now conclude thanks to Proposition~\ref{prop_interp_general}.

To show that \emph{(i)} and \emph{(i')} are equivalent, and that \emph{(ii')} implies \emph{(ii)},
use the long exact sequence in cohomology associated with the short
exact sequence
\beq
0\longto\cJ\cL\longto\cL\longto\cL|_{V(\cJ)}\longto0
\eeq
with $\cJ=\cI$ or $\cI'$, and $\cL=\cL_1$, $\cL_2$, or $\cL_1\tens\cL_2$.
\end{proof}

Remark that conditions \emph{(i)} and \emph{(ii)},
or \emph{(i')} and \emph{(ii')}, in Proposition~\ref{schema_arbitraire},
are very similar to conditions used to estimate the parameters
(dimension, distance) of AG codes.
Thus, borrowing techniques from this field,
one could hope to get good interpolation systems from classes
of varieties on which one knows how to construct good codes,
for example, algebraic surfaces, or toric varieties.

However up to now, the geometric objects that are best understood
from this point of view, especially regarding asymptotic properties,
are algebraic curves.
Thus interpolation systems constructed from algebraic curves will
be studied in the next section.

But before doing that, we give an example of use of the general
Proposition~\ref{schema_arbitraire}.

\begin{example}
It is well known that $\F_8$ admits a symmetric multiplication
algorithm of length $6$ over $\F_2$. This is best shown by giving
an explicit ad hoc description of this algorithm. It turns out
that this construction admits a nice interpretation in terms
of interpolation on the projective plane $\PP^2$ over $\F_2$.

So let $X=\PP^2$, and $\cL_1=\cL_2=\cO(1)$ the universal line bundle on it.
Let $x,y,z$ be the standard basis of $\Gamma(\PP^2,\cO(1))$, that is,
$x,y,z$ are the usual projective coordinate functions on $\PP^2$.

Write $\F_8=\F_2[\alpha]$ with $\alpha^3=\alpha+1$, and let $\Sigma'$
be (the schematic image of) the point with homogeneous coordinates
$(1:\alpha:\alpha^2)$. Hence evaluation at $\Sigma'$ maps the function
 $\lambda x+\mu y +\nu z\in\Gamma(\PP^2,\cO(1))$ to the element
$\lambda+\mu\alpha+\nu\alpha^2\in\F_8$, so the surjectivity condition \emph{(ii)}
in Proposition~\ref{schema_arbitraire} is satisfied (with, in fact, bijectivity).

For $\Sigma$ we choose the union of the six points
$(1:0:0)$ $(0:1:0)$ $(0:0:1)$ $(1:1:0)$ $(1:0:1)$ $(0:1:1)$,
and remark that evaluation of the basis functions
$x^2,y^2,z^2,xy,xz,yz$ of $\Gamma(\PP^2,\cO(2))$ at these six points
gives a triangular unipotent matrix, so the injectivity condition \emph{(i)}
is also satisfied (with, in fact, bijectivity).

This is enough to conclude the existence of the algorithm, but in fact,
since all proofs are constructive, we can describe it explicitly.
Write down the four evaluation maps
\begin{equation*}
\begin{array}{cccc}
f:&\Gamma(\PP^2,\cO(1))=<x,y,z>&\longto&\Gamma(\Sigma,\cO(1))\simeq(\F_2)^6\\
f':&\Gamma(\PP^2,\cO(1))=<x,y,z>&\longto&\Gamma(\Sigma',\cO(1))\simeq\F_8\\
F:&\Gamma(\PP^2,\cO(2))=<x^2,y^2,z^2,xy,xz,yz>&\longto&\Gamma(\Sigma,\cO(2))\simeq(\F_2)^6\\
F':&\Gamma(\PP^2,\cO(2))=<x^2,y^2,z^2,xy,xz,yz>&\longto&\Gamma(\Sigma',\cO(2))\simeq\F_8
\end{array}
\end{equation*}
where we have just seen that $f'$ and $F$ are bijective.
Now the proof of Proposition~\ref{prop_interp_general}
shows that multiplication in $\F_8$ decomposes as
\beq
\begin{CD}
\F_8\times\F_8 @>m_{\F_8}>> \F_8 \\
@V{\phi\times\phi}VV @AAwA\\
(\F_2)^6\times(\F_2)^6 @>m_{(\F_2)^6}>> (\F_2)^6
\end{CD}
\eeq
where $m_{(\F_2)^6}$ is coordinatewise multiplication, and $\phi=f\circ(f')^{-1}$
and $w=F'\circ F^{-1}$ are given in matrix form by
\beq
\phi=\left(\begin{array}{ccc}
1&0&0\\
0&1&0\\
0&0&1\\
1&1&0\\
1&0&1\\
0&1&1
\end{array}\right)
\qquad\quad
w=\left(\begin{array}{cccccc}
1&1&1&0&0&1\\
1&0&0&1&0&1\\
1&1&0&0&1&0
\end{array}\right)
\eeq
relative to the basis $1,\alpha,\alpha^2$ of $\F_8$ and the canonical
basis of $(\F_2)^6$, with column vector convention.

Of course there are other ways to interpret this construction,
for example, as interpolation on the affine space $\mathbb{A}^3$.
However remark that this would not have been possible working
with \emph{curves} only (or at least, not in a natural way),
because curves over $\F_2$ of sufficiently small genus
do not admit enough points for the interpolation to be possible. 
\end{example}

Another situation in which Proposition~\ref{schema_arbitraire}
could be useful is if one is interested in the bilinear complexity
of a local algebra $\cA'$ that cannot
be written as a quotient of a polynomial algebra in only one variable.
Indeed such an algebra cannot be ``embedded'' in a curve
(see the discussion on monogenous local algebras below), hence
requires higher-dimensional objects for interpolation.

\section{The extended Chudnovsky-Chudnovsky algorithm}

\label{section3}

From now on, $K$ will be a finite field, say $K=\F_q$.
We will only consider algebras that are associative,
commutative, and with unity.

In fact we will be particularly interested in the following family
of $\F_q$-algebras, and their bilinear complexities:

\begin{definition}
For any integers $m,l\geq1$ we consider the $\F_q$-algebra
%of polynomials, or here equivalently of power series, in one indeterminate,
%with coefficients in $\F_{q^m}$, truncated at order~$l$:
of polynomials in one indeterminate
with coefficients in $\F_{q^m}$, truncated at order~$l$:
\beq
%\cA_q(m,l)=\F_{q^m}[t]/(t^l)=\F_{q^m}[[t]]/(t^l)
\cA_q(m,l)=\F_{q^m}[t]/(t^l)
\eeq
of dimension
\beq
\dim_{\F_q}\cA_q(m,l)=ml,
\eeq
and we denote by
\beq
\mu_q(m,l)=\mu(\cA_q(m,l)/\F_q)
\eeq
its bilinear complexity over $\F_q$.
\end{definition}

Of special significance are the following two cases:
when $l=1$,
\beq
\mu_q(m,1)=\mu_q(m)
\eeq
is the bilinear complexity of multiplication in $\F_{q^m}$ over $\F_q$;
and when $m=1$,
\beq
\mu_q(1,l)=\widehat{M}_q(l)
\eeq
is the quantity used in the estimates of~\cite{CO}.

\begin{lemma}
\label{mu(m,l)<mu(m)M(l)}
With the notations above,
\beq
%\mu_q(m,l)\leq\mu_q(m)\widehat{M}_q(l).
\mu_q(m,l)\leq\mu_q(m)\widehat{M}_{q^m}(l).
\eeq
\end{lemma}
\begin{proof}
%Direct consequence of Lemma~\ref{functorial_inequalities}~(iv)
%(or also Lemma~\ref{functorial_inequalities} (i)-(ii)).
Direct consequence of Lemma~\ref{functorial_inequalities}.b).
\end{proof}

\begin{remark}
As will be shown later, there are examples where this inequality is strict.
\end{remark}

We now introduce another class of $\F_q$-algebras,
before studying how they relate to the $\cA_q(m,l)$:
\begin{itemize}
\item We say that a finite-dimensional $\F_q$-algebra $A$ is \emph{monogenous}
if it can be written as a quotient of the ring of polynomials in one
indeterminate over~$\F_q$, say: $A\simeq\F_q[t]/(P(t))$.
These are precisely the algebras whose bilinear complexity was first
studied in \cite{FZ,Winograd}.

Moreover we say that $A$ is \emph{local} if it has only one maximal ideal.
Thus, by the Chinese remainder theorem, a monogenous local $\F_q$-algebra
is necessarily of the form
\beq
A\simeq\F_q[t]/(Q(t)^l)
\eeq
for some \emph{irreducible} polynomial $Q$ over $\F_q$
and some integer $l\geq1$.
\item More generally, let $X$ be an algebraic curve over $\F_q$
(the situation discussed just above corresponds to the case $X=\PP^1$).
By a \emph{thickened point} in $X$ we mean any closed subscheme of $X$
supported on a closed point (of arbitrary degree).
For example, if $Q$ is a closed point in $X$, we denote by $\cI_Q$
the sheaf of ideals defining it, and for any integer $l\geq1$ we
let $Q^{[l]}$ be the closed subscheme of $X$ defined by the sheaf
of ideals $(\cI_Q)^l$. Then $Q^{[l]}$ is a thickened point supported on $Q$.
Conversely, any thickened point in $X$ is of this form.
Indeed, by convention a curve $X$ is always supposed smooth,
hence the local ring $\cO_{X,Q}$ of $X$ at $Q$ is principal,
and every ideal in this ring is of the form $(t_Q^{\;l})$,
where $t_Q$ is a local parameter at $Q$.

We remark that such a thickened point is necessarily affine, and we let
\beq
\cA_{Q^{[l]}}=\Gamma(Q^{[l]},\cO_{Q^{[l]}})=\Gamma(X,\cO_X/(\cI_Q)^l)=\cO_{X,Q}/(t_Q^{\;l})
\eeq
be its ring of regular functions.
\end{itemize}

\begin{lemma}
\label{Cohen}
Any monogenous local $\F_q$-algebra,
and more generally the ring of functions of any thickened point on a
curve over $\F_q$, is isomorphic to some $\cA_q(m,l)$. More precisely:
\begin{itemize}
\item Let $Q$ be an irreducible polynomial over $\F_q$,
of degree $\deg Q=m$, and let $l\geq1$ be an integer.
Then, as $\F_q$-algebras,
\beq
\F_q[t]/(Q(t)^l)\simeq\cA_q(m,l).
\eeq
\item More generally,
let $X$ be a curve over $\F_q$ and $Q$ a closed point in $X$,
of degree $\deg Q=m$, and let $l\geq1$ be an integer.
Then, as $\F_q$-algebras,
\beq
\cA_{Q^{[l]}}\simeq\cA_q(m,l).
\eeq
\end{itemize}
As a consequence, all these algebras have the same bilinear
complexity $\mu_q(m,l)$.
\end{lemma}
\begin{proof}
This is a special case of Cohen's structure theorem for complete local
rings in equal characteristic
(see e.g. \cite{BourbakiAC9} AC~IX.30, \S3, Th.~2).
But for ease of the reader we recall how this works concretely
in our specific situation.

Write $\cA_{Q^{[l]}}=\cO_{X,Q}/(t_Q^{\;l})$, where $\cO_{X,Q}$ is the local ring
of $X$ at $Q$, and $t_Q$ a local parameter.
We will construct an isomorphism
\beq
\label{isom_O/t[t]/t^l}
(\cO_{X,Q}/t_Q)[t]/(t^l)\overset{\sim}{\longto}\cO_{X,Q}/(t_Q^{\;l})
\eeq
hence proving the lemma, since $\cO_{X,Q}/(t_Q)\simeq\F_{q^m}$.

To do so, first choose any $\alpha$ generating $\cO_{X,Q}/(t_Q)$ over $\F_q$,
with minimal polynomial $F_\alpha$, and invoke Hensel's lemma
to lift $\alpha$ to $\widetilde{\alpha}$ root of $F_\alpha$
in $\cO_{X,Q}/(t_Q^{\;l})$. Sending $\alpha$ to $\widetilde{\alpha}$
then defines
a morphism of $\F_q$-algebras
\beq
\label{Hensel_O/t}
\cO_{X,Q}/(t_Q)\longto\cO_{X,Q}/(t_Q^{\;l})
\eeq
section of the natural
projection $\cO_{X,Q}/(t_Q^{\;l})\longto\cO_{X,Q}/(t_Q)$,
and to conclude, we extend \eqref{Hensel_O/t} to \eqref{isom_O/t[t]/t^l}
by sending $t$ to $t_Q$.
\end{proof}

%If $X$ is an algebraic curve over $\F_q$,
%by a divisor on $X$ we will always mean a divisor that is
%\emph{$\F_q$-rational}.
%If $D$ is such a divisor on $X$, we denote by
If $X$ is an algebraic curve over $\F_q$,
and $D$ a divisor on $X$, we denote by
\beq
L(D)=\Gamma(X,\cO_X(D))
\eeq
its Riemann-Roch space,
and by
\beq
l(D)=\dim L(D)
\eeq
the dimension (over $\F_q$) of the latter.
We also choose a canonical divisor $K_X$ on $X$
and we let
\beq
i(D)=l(K_X-D)
\eeq
be the \emph{index of specialty} of $D$.
Recall that the Riemann-Roch theorem can then be stated as
\beq
l(D)-i(D)=\deg D+1-g
\eeq
where $g$ is the genus of $X$.

\begin{theorem}
\label{ChCh+gen}
Let $X$ be a curve of genus $g$ over $\F_q$, and let $m,l\geq1$ be two integers.
Suppose that $X$ admits a closed point $Q$ of degree $\deg Q=m$.
Let $G$ be an effective divisor on $X$, and write
\beq
G=u_1P_1+\cdots+u_nP_n
\eeq
where the $P_i$ are pairwise distinct closed points,
of degree $\deg P_i=d_i$.
Suppose there exist two divisors $D_1,D_2$ on $X$ such that:
\begin{enumerate}[(i)]
\item The natural evaluation map
\beq
L(D_1+D_2)\longto \prod_{i=1}^n\cO_X(D_1+D_2)|_{P_i^{[u_i]}}
\eeq
is injective.
\item The natural evaluation maps
\beq
L(D_1)\longto \cO_X(D_1)|_{Q^{[l]}}\qquad L(D_2)\longto \cO_X(D_2)|_{Q^{[l]}}
\eeq
are surjective.
\end{enumerate}
Then
\beq
\label{mu(m,l)<mu(d,u)}
\mu_q(m,l)\leq\sum_{i=1}^n\mu_q(d_i,u_i).
\eeq
In fact we also have
$\mu_q(m,l)\leq\mu(\prod_{i=1}^n\cA_q(d_i,u_i)/\F_q)$.
Moreover, if $D_1=D_2$, all these inequalities also hold for the symmetric
bilinear complexity $\mu^{\textrm{sym}}$.

Sufficient numerical
criteria for the hypotheses above to hold can be given as follows.
A sufficient condition for the existence of $Q$ of degree $m$ on $X$ is
that $2g+1\leq q^{(m-1)/2}(q^{1/2}-1)$, while sufficient conditions
for (i) and (ii) are:
\begin{enumerate}[(i')]
\item The divisor $D_1+D_2-G$ is zero-dimensional:
\beq
l(D_1+D_2-G)=0.
\eeq
\item The divisors $D_1-lQ$ and $D_2-lQ$ are non-special:
\beq
i(D_1-lQ)=i(D_2-lQ)=0.
\eeq
\end{enumerate}
More precisely,
(i) and (i') are equivalent, while (ii') only implies (ii)
a priori.
\end{theorem}
\begin{proof}
Use Proposition~\ref{schema_arbitraire} with
$\Sigma=P_1^{[u_1]}\cup\cdots\cup P_n^{[u_n]}$,
$\Sigma'=Q^{[l]}$,
and $\cL_1=\cO_X(D_1)$ and $\cL_2=\cO_X(D_2)$.
Combined with Lemma~\ref{Cohen} this gives
\beq
\label{mu(m,l)<mu(P(A(d,u)))}
\mu_q(m,l)\leq\mu(\prod_{i=1}^n\cA_q(d_i,u_i)/\F_q)
\eeq
as claimed.
One can then apply
Lemma~\ref{functorial_inequalities}.c)
to get \eqref{mu(m,l)<mu(d,u)} (whether we lose
in passing from \eqref{mu(m,l)<mu(P(A(d,u)))} to \eqref{mu(m,l)<mu(d,u)}
depends on the direct sum conjecture \eqref{direct_sum_conj}).

As for the numerical sufficient condition stated here for the existence of $Q$,
it can be found in \cite{Stichtenoth}, Cor.~V.2.10.(c).
\end{proof}

\begin{remark}
For applications it might be useful to make things more explicit,
so we describe in more concrete terms how the various geometric data
in Theorem~\ref{ChCh+gen} lead to an interpolation system as in
Definition~\ref{def_interpol}. The key point is to describe the
evaluation maps, which can be done in relatively elementary terms
when $X$ is a curve. For example we describe the composite map
\beq
L(D_1)\longto\cO_X(D_1)|_{Q^{[l]}}\overset{\sim}{\longto}\cA_q(m,l).
\eeq
As a first step, we choose a local parameter $t_Q$ at $Q$.
Then $t_Q^{v_Q(D_1)}$ is a local generator for $\cO_X(D_1)$ at $Q$,
and we use this local generator to define a trivialization
$\cO_X(D_1)|_{Q^{[l]}}\simeq\cO_X|_{Q^{[l]}}=\cO_{X,Q}/(t_Q^{\;l})$
as asked in~\eqref{basic_triv}. Thus we get a map
\beq
\label{explicit_eval}
\begin{array}{ccc}
L(D_1) & \longto & \cO_{X,Q}/(t_Q^{\;l}) \\
f & \mapsto & t_Q^{-v_Q(D_1)}f \mod (t_Q^{\;l})
\end{array}
\eeq
and we compose this with the isomorphism
$\cO_{X,Q}/(t_Q^{\;l})\overset{\sim}{\longto}(\cO_{X,Q}/(t_Q))[t]/(t^l)\simeq\cA_q(m,l)$
given in Lemma~\ref{Cohen} (and explicited in its proof) to conclude.

The other maps $L(D_2)\longto\cA_q(m,l)$
and $L(D_1+D_2)\longto\prod_{i=1}^n\cA_q(d_i,u_i)$ are described
in the same way.

A nice property of these evaluation maps,
as is best seen from \eqref{explicit_eval},
is that they do not need the points at which we evaluate
to be disjoint from the support of the divisor
(although this is not a crucial point of the construction,
since this situation can also be avoided thanks to the
strong approximation theorem).
\end{remark}

\begin{remark}
This Theorem~\ref{ChCh+gen} encompasses essentially all presently known
variants of the Chudnovsky-Chudnovsky interpolation method as special cases.
For example, restricting to $l=1$ and $D_1=D_2$,
and using Lemma~\ref{mu(m,l)<mu(m)M(l)}, gives Th.~3.1 of \cite{CO}
(if one further restricts to all $d_i=u_i=1$, this gives the
original version of Chudnovsky-Chudnovsky \cite{ChCh}).
Thus one can say that Theorem~\ref{ChCh+gen} improves the method
of \cite{CO} in at least two points:
\begin{itemize}
\item Allowing asymmetry ($D_1\neq D_2$) makes conditions \emph{(i)} and \emph{(ii)},
or \emph{(i')} and \emph{(ii')}, easier to satisfy than their counterparts
in \cite{CO}; in turn this allows more flexibility in the choice of
the curve $X$ and the divisor $G$.
\item The use of $\mu_q(d,u)$
in the right-hand side of~\eqref{mu(m,l)<mu(d,u)},
instead of $\mu_q(d)\widehat{M}_{q^{d}}(u)$ as in \cite{CO},
leads to stronger estimates. Of course, for this to be useful,
one needs upper bounds on these $\mu_q(d,u)$ that are better
than the one given in Lemma~\ref{mu(m,l)<mu(m)M(l)}.
But a nice feature of~\eqref{mu(m,l)<mu(d,u)} is that
this same quantity $\mu_q(m,l)$ also appears in the left-hand side,
so we can try to get these
upper bounds from Theorem~\ref{ChCh+gen} itself,
in a sort of recursive procedure.
\end{itemize}
These points will be illustrated in the following three sections.
\end{remark}

\section{Genus $0$ or $1$}

\label{section4}

The main motivation for this section is the following:

\begin{question}
What is the actual value of $\mu_q(m,l)$ for small $q,m,l$?
Or at least, find upper bounds that are better than the one given
in Lemma~\ref{mu(m,l)<mu(m)M(l)}.
\end{question}

Answering this question can lead to improved bounds also for high values
of the parameters. For example, suppose that in Theorem~\ref{ChCh+gen}
we take $l=1$ and
the divisor $G$ consists of:
\begin{itemize}
\item $N_1$ points of degree $1$, of which $l_1$ with multiplicity $2$ and
the remaining $N_1-l_1$ with multiplicity $1$
\item $N_2$ points of degree $2$, of which $l_2$ with multiplicity $2$ and
the remaining $N_2-l_2$ with multiplicity $1$
\item $N_4$ points of degree $4$, of which $l_4$ with multiplicity $2$ and
the remaining $N_4-l_4$ with multiplicity $1$.
\end{itemize}
Then \eqref{mu(m,l)<mu(d,u)} gives
\beq
\mu_q(m)\leq N_1+2l_1+3N_2+(\mu_q(2,2)-3)l_2+\mu_q(4)N_4+(\mu_q(4,2)-\mu_q(4))l_4.
\eeq
Provided $\mu_q(2,2)<9$ or $\mu_q(4,2)<3\mu_q(4)$,
this improves the bound in Prop.~3.1 of \cite{BP}.
Such bounds on $\mu_q(2,2)$ or $\mu_q(4,2)$ will be given
in Examples~\ref{ex22} and~\ref{ex42} and Lemma~\ref{lem42} below.

\begin{proposition}
\label{g=0}
Let $m,l\geq1$ be two integers
with
\beq
ml\leq\frac{q}{2}+1.
\eeq
Then
\beq
{}_{\phantom{\substack{\Sigma\\ \Sigma}}}\mu_q(m,l)\leq\mu^{\textrm{sym}}_q(m,l)\leq2ml-1._{\phantom{\substack{\Sigma\\ \Sigma}}}
\eeq
More generally let $G$ be an effective divisor on $\PP^1$, and write
\beq
G=u_1P_1+\cdots+u_nP_n
\eeq
where the $P_i$ are pairwise distinct closed points,
of degree $\deg P_i=d_i$.
Suppose
\beq
\deg G=\sum_{i=1}^nd_iu_i\geq 2ml-1.
\eeq
Then
\beq
\mu_q(m,l)\leq\sum_{i=1}^n\mu_q(d_i,u_i)
\eeq
and likewise $\mu^{\textrm{sym}}_q(m,l)\leq\sum_{i=1}^n\mu^{\textrm{sym}}_q(d_i,u_i)$.
\end{proposition}
\begin{proof}
Remark that the first assertion is a particular case of the second,
because if $n=2ml-1\leq q+1$,
we can find $n$ distinct points
of degree $1$ on $\PP^1$ and let $G$ be their sum.
Recall also that $\PP^1$ admits points of any degree,
and that any divisor of degree $-1$ on $\PP^1$ is both
zero-dimensional and non-special.
So, to conclude, let $D$ be any divisor of degree $ml-1$ on $\PP^1$,
and apply Theorem~\ref{ChCh+gen} with $D_1=D_2=D$.
\end{proof}

Recall that an elliptic curve over $\F_q$
is a curve $X$ of genus $1$ with a chosen
point $P_\infty\in X(\F_q)$.
This set $X(\F_q)$ of $\F_q$-rational points of $X$,
or equivalently, of closed points of degree~$1$,
then admits a structure of abelian group
with identity element $P_\infty$.
Also, given such an elliptic curve, there is a map
\beq
\sigma:\Div(X)\longto X(\F_q)
\eeq
uniquely defined by the condition that each divisor $D$ of degree $d$
is linearly equivalent to the divisor $\sigma(D)+(d-1)P_\infty$.
This map $\sigma$ is a group morphism,
it passes to linear equivalence, and induces an isomorphism
of the degree $0$ class group $\Cl^0(X)$ with $X(\F_q)$.
We now generalize a result of Shokrollahi \cite{Shokro92}
and Chaumine~\cite{Chaumine}:

\begin{proposition}
\label{g=1}
Let $X$ be an elliptic curve over $\F_q$, with all notations as above.
Let $m,l\geq1$ be two integers.
Suppose that $X$ admits a closed point $Q$ of degree $\deg Q=m$.
Let $G$ be an effective divisor on $X$, and write
\beq
G=u_1P_1+\cdots+u_nP_n
\eeq
where the $P_i$ are pairwise distinct closed points,
of degree $\deg P_i=d_i$,
so $\deg G=\sum_{i=1}^nd_iu_i$.
Then
\beq
\mu_q(m,l)\leq\sum_{i=1}^n\mu_q(d_i,u_i)
\eeq
provided one of the following conditions is satisfied:
\begin{enumerate}[a)]
\item
$\deg G=2ml$ and $|X(\F_q)|\geq3$
\item
$\deg G=2ml$ and $|X(\F_q)|\geq2$, and either $\sigma(G)\neq P_\infty$
or $X(\F_q)$ is not entirely of $2$-torsion (or both)
\item
$\deg G\geq 2ml+1$
and $|X(\F_q)|\geq2$
\item
$\deg G\geq 2ml+3$.
\end{enumerate}
Moreover in cases b), c), or d), one also has
$\mu_q^{\textrm{sym}}(m,l)\leq\sum_{i=1}^n\mu_q^{\textrm{sym}}(d_i,u_i)$.
\end{proposition}
\begin{proof}
Recall that a divisor of degree $0$ on $X$ is both
zero-dimensional and non-special, unless it is linearly equivalent
to zero.
%Recall also that the (degree $0$) divisor class group of $X$ is isomorphic
%to the group of points $X(\F_q)$. 

Suppose first we're in case a), so $X(\F_q)\simeq\Cl^0(X)$
has order at least $3$.
%Since $\sigma$ identifies $\Cl^0(X)$ with $X(\F_q)$,
This implies that
there are two divisors $Z$ and $Z'$ of degree $0$ on $X$
that are not linearly equivalent nor linearly equivalent to zero.
Let then $D_1=lQ+Z$, and let $D_2=lQ+Z$ or $lQ+Z'$, depending on whether
$D_1+D_2-G=2lQ+2Z-G$ or $2lQ+Z+Z'-G$ is not linearly equivalent to zero.
With this choice, conditions \emph{(i')} and \emph{(ii')} in Theorem~\ref{ChCh+gen}
are satisfied, and the conclusion follows.

Suppose now we're in case b).
Suppose first that $X(\F_q)\simeq\Cl^0(X)$ is not
entirely of $2$-torsion. Then there are two divisors $Z$ and $Z'$ of
degree $0$ not linearly equivalent to zero, and such that $2Z$ and $2Z'$
are not linearly equivalent. Let then $D_1=D_2=lQ+Z$ or $D_1=D_2=lQ+Z'$,
depending on whether
$D_1+D_2-G=2lQ+2Z-G$ or $2lQ+2Z'-G$ is not linearly equivalent to zero.
With this choice, conditions \emph{(i')} and \emph{(ii')} are
satisfied again.
On the other hand, suppose $X(\F_q)\simeq\Cl^0(X)$ is entirely of $2$-torsion,
so that $\sigma(G)\neq P_\infty$ by our hypothesis.
Let $Z$ be a divisor of degree $0$ not linearly equivalent to zero
(it exists since $|X(\F_q)|\geq2$) and take $D_1=D_2=lQ+Z$,
so condition \emph{(ii')} is satisfied.
Then $\sigma(D_1+D_2-G)=\sigma(G)\neq P_\infty$
and condition \emph{(i')} is also satisfied.

%If $X(\F_q)$ has order at least $2$,
%there exists a divisor $Z$ of degree $0$ on $X$
%not linearly equivalent to zero. We take $D_1=D_2=lQ+Z$ so condition \emph{(ii')}
%is satisfied.
%But then, in case b), $\Cl^0(X)$ is of $2$-torsion, so $\sigma(D_1+D_2-G)=\sigma(G)$,
%and condition \emph{(i')} is satisfied by our hypothesis $\sigma(G)\neq P_\infty$. While in case c), condition \emph{(i')} is also satisfied,
%for degree reasons.

Case c) works likewise: let $Z$ be a divisor of degree $0$
not linearly equivalent to zero and take $D_1=D_2=lQ+Z$,
so condition \emph{(ii')} is satisfied,
while condition \emph{(i')} is also satisfied for degree reasons.

In case d),
we take $D_1=D_2=(ml+1)P_\infty$. Then conditions \emph{(i')} and \emph{(ii')}
are satisfied
for degree reasons.

Last, remark that except perhaps in case a), we always took $D_1=D_2$ in the
proof, so that the estimates then also work for the symmetric bilinear
complexity $\mu^{\textrm{sym}}$.
\end{proof}

\begin{example}
\label{ex22}
Proposition~\ref{g=0} gives
\beq
\mu_q(2,2)\leq 7\qquad\textrm{for $q\geq7$}
\eeq
and Proposition~\ref{g=1} gives
\beq
\mu_q(2,2)\leq 8\qquad\textrm{for $q=4$ or $5$}.
\eeq
Indeed, recall that the number of points of degree $1$ on an elliptic
curve $X$ over $\F_q$ can be written as $|X(\F_q)|=q+1-t$
for some integer $t$, the \emph{trace} of $X$, satisfying $|t|\leq2\sqrt{q}$.
Conversely, Honda-Tate theory gives additional sufficient and necessary
conditions on $t$ for a curve having this number of points to exist
(\cite{Waterhouse}, Th.~4.1). The trace $t$ then also determines
the number of points on $X$ of any degree. For example,
we have $|X(\F_{q^2})|=(q+1)^2-t^2$,
hence $X$ has $\frac{1}{2}(|X(\F_{q^2})|-|X(\F_q)|)=\frac{(q+1-t)(q+t)}{2}$
points of degree $2$ (and likewise, $\frac{((q+1)^2-t^2)(q^2-2q+t^2)}{4}$
points of degree $4$, we will use it in the next example).

Using this machinery,
we see that for $q=4$ or $5$ there exists an elliptic curve over $\F_q$ with
eight points of degree $1$ (and at least one point of degree $2$), so
in Proposition~\ref{g=1} we can take as $G$ all these points of
degree $1$, each with multiplicity $1$.

Unfortunately it seems difficult to improve the bound $\mu_q(2,2)\leq 9$
for $q=2$ or $3$, at least with this generic method. Whether this is the
exact value is yet unsettled.
\end{example}

\begin{example}
\label{ex42}
Proposition~\ref{g=0} gives
\beq
\mu_q(4,2)\leq 15\qquad\textrm{for $q\geq16$}
\eeq
and Proposition~\ref{g=1} gives
\beq
\mu_q(4,2)\leq 16\qquad\textrm{for $q=9$, $11$, or $13$}
\eeq

\beq
\mu_8(4,2)\leq 17\qquad\mu_7(4,2)\leq 18\qquad\mu_5(4,2)\leq 19
\eeq

\beq
\mu_4(4,2)\leq 20\qquad
\mu_3(4,2)\leq 23\qquad
\mu_2(4,2)\leq 26.
\eeq
The proof of these bounds follows the same lines as in the previous example.

For $q=9$, $11$, or $13$, there is an elliptic curve over $\F_q$
with $16$ points of degree $1$ (and at least one point of degree $4$), so
in Proposition~\ref{g=1} we can take as $G$ all these points of
degree $1$, each with multiplicity $1$.

For $q=8$ we can choose the trace $t=-5$,
and $G$ consists of $14$ points of degree $1$
and $1$ point of degree $2$, all with multiplicity $1$.

For $q=7$ we choose $t=-5$, and $G$ consists of $12$ points of degree $1$
and $2$ points of degree $2$, all with multiplicity $1$.

For $q=5$ we choose $t=-4$, and $G$ consists of $10$ points of degree $1$
and $3$ points of degree $2$, all with multiplicity $1$.

For $q=4$ we choose $t=-3$, and $G$ consists of $8$ points of degree $1$
and $4$ points of degree $2$, all with multiplicity $1$.

For $q=3$ we choose $t=-2$, and $G$ consists of
$2$ points of degree $1$ with multiplicity $1$,
$4$ points of degree $1$ with multiplicity $2$,
and $3$ points of degree $2$ with multiplicity $1$.

For $q=2$ we choose $t=-1$, and $G$ consists of
$4$ points of degree $1$ with multiplicity $3$,
and $2$ points of degree $2$ with multiplicity $1$.

Remark that all these bounds already improve the one given by
Lemma~\ref{mu(m,l)<mu(m)M(l)} (at least given the best
upper bounds on $\mu_q(4)$ known up to now).
However, for small $q$ it is possible to do even better as follows.
\end{example}

\begin{lemma}
\label{lem42}
Suppose $m$ is not prime, and write $m=de$ for some integers $d,e\geq2$.
Then
\beq
\mu_q(m,l)\leq\mu_q(d)\mu_{q^d}(e,l)^{\phantom{\Sigma}}_{\phantom{\substack{\Sigma\\ \Sigma}}}
\eeq
(and likewise $\mu_q^{\textrm{sym}}(m,l)\leq\mu_q^{\textrm{sym}}(d)\mu_{q^d}^{\textrm{sym}}(e,l)$).
In particular:
\beq
\mu_3(4,2)\leq\mu_3(2)\mu_9(2,2)\leq 21,\qquad \mu_2(4,2)\leq\mu_2(2)\mu_4(2,2)\leq 24.
\eeq
\end{lemma}
\begin{proof}
Direct consequence of Lemma~\ref{functorial_inequalities}.b),
noting that $\cA_q(m,l)$ can be considered as an algebra over $\F_{q^d}$,
and as such can be identified with $\cA_{q^d}(e,l)$.
\end{proof}

We do not claim these new upper bounds to be optimal. Any further improvement
(as well as lower bounds, on the other side)
would be of interest.

\begin{example}
In \cite{CO}, section~5, Cenk and \"Ozbudak give upper bounds on $\mu_2(163)$
and $\mu_3(97)$. However there is an error in their proof of the first,
and the second would need a slight extra justification.

The origin of the error is in their Th.~3.6,
which, as stated, is false. Condition~(1) in this Th.~3.6
asks for the existence of a non-special divisor of degree $n+g-1$ 
(instead of $g-1$ as in their Th.~3.2 or Cor.~3.5)
in order for their evaluation map $Ev_Q$
to be surjective. However this condition is not sufficient,
as illustrated as follows.

To give an upper bound on $\mu_2(163)$, the authors of \cite{CO}
introduce the elliptic curve $y^2+y=x^3+x+1$ over $\F_2$,
which has only one point of degree $1$, which means that its class
group $\Cl^0$ is trivial. They take a point $Q$ of degree $163$ on this curve,
and a non-special divisor $D$ of degree $163$ disjoint from $Q$.
They need their map $Ev_Q:L(D)\longto\cO_Q/Q$ to be surjective
(which the proof of their Th.~3.6 claims).
However, this map fits in the long exact sequence
\beq
0\longto L(D-Q)\longto L(D)\longto \cO_Q/Q\longto \dots
\eeq
and since $D-Q$ has degree $0$, and the curve has
trivial class group, we have $D-Q\sim 0$ and $l(D-Q)=1$.
This means that $Ev_Q$ is non-injective, and since $L(D)$ and $\cO_Q/Q$
have the same dimension (namely $163$), $Ev_Q$ is non-surjective as well.

To fix this error, we can use our Proposition~\ref{g=1} instead.
We use the same curve as in \cite{CO}, but since this curve
has only one point of degree $1$, we need case d) of the proposition,
and the divisor $G$ has to be modified accordingly: we take the only
point of degree $1$ with multiplicity $5$,
and then we take
all $2$ points of degree $2$,
all $4$ points of degree $3$,
all $5$ points of degree $4$,
all $8$ points of degree $5$,
all $8$ points of degree $6$,
all $25$ points of degree $8$,
all with multiplicity $1$. 
Then $G$ has degree
\beq
\deg G=1\cdot5+2\cdot2+4\cdot3+5\cdot4+8\cdot5+8\cdot6+25\cdot8=329=2\cdot163+3\;
\eeq
and Proposition~\ref{g=1}.d) gives
\beq
\begin{split}
\mu_2(163)\leq \mu_2(1,5)+2\mu_2(2)+&4\mu_2(3)+5\mu_2(4)+\\
&+8\mu_2(5)+8\mu_2(6)+25\mu_2(8)\leq 910.
\end{split}
\eeq
See \cite{CO}, Table~1, for the numerical details.
Remark they give the upper bound $\mu_2(7)\leq22$,
with the quotient $22/7$ being the highest among similar
estimates up to degree~$8$.
This is why we didn't use points of degree $7$ in our $G$,
and explains
why our upper bound $910$ is better than the upper bound $916$
in \cite{CO},
despite our $G$ having higher degree. This said, perhaps further optimizations
of this sort are possible.

Concerning the upper bound $\mu_3(97)\leq 426$, Cenk and \"Ozbudak use the
curve $y^2=x^3+x^2+2x+1$ over $\F_3$.
This curve has $3$ points of degree $1$, hence
its $\Cl^0$ is non-trivial, so the error in Condition~(1) of their
Th.~3.6 is not harmful. However for their upper bound to be fully
justified they also need to explain why their application $\phi$ is injective,
which they do not. But here again we can use Proposition~\ref{g=1} (case~a)
instead, with the same curve and the same divisor $G$ as theirs.
This gives the same bound $\mu_3(97)\leq 426$, without needing
any extra justification.
\end{example}

\section{Fixing some bounds of Ballet}

For any curve $X$ over $\F_q$, we denote by $B_d(X/\F_q)$
the number of closed points of degree $d$ on $X$, so that, for all $n$,
\beq
|X(\F_{q^n})|=\sum_{d|n}dB_d(X/\F_q).
\eeq

We now want to apply Theorem~\ref{ChCh+gen} with curves of higher genus,
as well as give easy verifiable criteria for the existence of divisors
$D_1,D_2$ satisfying conditions \emph{(i)} and \emph{(ii)},
or \emph{(i')} and \emph{(ii')}, in
this theorem.
%More precisely, 
For example, we can do so as these conditions be satisfied
for degree reasons:

\begin{proposition}
Let $X$ be a curve of genus $g$ over $\F_q$, and let $m,l\geq1$ be two integers.

Suppose that $X$ admits a closed point $Q$ of degree $\deg Q=m$
(a sufficient condition for this is $2g+1\leq q^{(m-1)/2}(q^{1/2}-1)$).

Suppose also that $X$ admits a non-special divisor $S$, of degree $g+e-1$,
for an integer $e$ as small as possible (hence $e\leq g$ by the Riemann-Roch theorem).

Consider now a collection of integers $n_{d,u}\geq0$
(for $d,u\geq1$), such that almost all of them are zero, and
that for any $d$,
\beq
\label{nd<Nd}
n_d=\sum_un_{d,u}\leq B_d(X/\F_q).
\eeq
Then, provided
\beq
\label{G>2n+2g+2e-1}
\sum_{d,u}n_{d,u}du\geq2ml+2e+2g-1
\eeq
we have
\beq
\mu_q(m,l)\leq\sum_{d,u}n_{d,u}\mu_q(d,u)
\eeq
and likewise
\beq
\mu_q^{\textrm{sym}}(m,l)\leq\sum_{d,u}n_{d,u}\mu_q^{\textrm{sym}}(d,u).
\eeq
\end{proposition}
\begin{proof}
For $1\leq j\leq n_{d,u}$ choose a point $P_{d,u,j}$ of degree $d$
in $X$, such that $P_{d,u,j}\neq P_{d,u',j'}$ if $(u,j)\neq(u',j')$.
This is possible by~\eqref{nd<Nd}.
Let then $G=\sum_{d,u}\sum_{1\leq j\leq n_{d,u}}uP_{d,u,j}$,
so that $\deg G=\sum_{d,u}n_{d,u}du$.
Let also $D=D_1=D_2=S+lQ$, so $D-lQ$ is non-special,
and $2D-G$ has negative degree by~\eqref{G>2n+2g+2e-1}.
Hence conditions \emph{(i')} and \emph{(ii')} in Theorem~\ref{ChCh+gen}
are satisfied and we can conclude.
\end{proof}

In order to use this proposition one needs good upper bounds on $e$.
For results of this type, see for example \cite{BL2006} or \cite{BRR2010}.
In many cases it is possible to take $e=0$.
However under some mild hypothesis on $q$ or $X$,
it is possible to do substantially better, namely we can
gain an additional constant $g$ in~\eqref{G>2n+2g+2e-1}.
For this to be possible, one needs to replace the degree argument
in the proof with a finer method ensuring that
conditions \emph{(i')} and \emph{(ii')} are still satisfied for some
divisors $D_1,D_2$ of appropriate degree.
Having allowed asymmetry in our interpolation system will make this easier.
In fact we will give two different methods achieving this.
The first one will show the existence of $D_1,D_2$ using a cardinality
argument. The second one will be more constructive, and works also
in a symmetric setting, although only under more restrictive conditions.

\begin{theorem}
\label{th=g-1}
Let $X$ be a curve of genus $g$ over $\F_q$, and let $m,l\geq1$ be two integers.

Suppose that $X$ admits a closed point $Q$ of degree $\deg Q=m$
(a sufficient condition for this is $2g+1\leq q^{(m-1)/2}(q^{1/2}-1)$).

Consider now a collection of integers $n_{d,u}\geq0$
(for $d,u\geq1$), such that almost all of them are zero, and
that for any $d$,
\beq
\label{nd<Ndbis}
n_d=\sum_un_{d,u}\leq B_d(X/\F_q).
\eeq
Suppose also
\beq
\label{G>2n+g-1}
\sum_{d,u}n_{d,u}du\geq2ml+g-1.
\eeq
Then:
\begin{enumerate}[a)]
\item
If $q>5$,
we have
\beq
\mu_q(m,l)\leq\sum_{d,u}n_{d,u}\mu_q(d,u).
\eeq
\item
If $|X(\F_q)|>2g$, we have
\beq
\mu_q(m,l)\leq\sum_{d,u}n_{d,u}\mu_q(d,u).
\eeq
Moreover,
suppose $X$ and $Q$ are given explicitly,
that $2g+1$ points of degree $1$ on $X$ are given explicitly,
and, for any $d$,
that $n_d$ points of degree $d$ on $X$ are given explicitly.
Suppose also that for each $d,u$ such that $n_{d,u}>0$,
we are given explicitly a multiplication algorithm of length $l_{d,u}$
for $\cA_q(d,u)$.
Then, after at most $3g^2$ computations of Riemann-Roch spaces on $X$,
we can construct explicitly a multiplication algorithm
of length $\sum_{d,u}n_{d,u}l_{d,u}$ for $\cA_q(m,l)$.
\item
If $|X(\F_q)|>5g$,
we have
\beq
\mu_q^{\textrm{sym}}(m,l)\leq\sum_{d,u}n_{d,u}\mu_q^{\textrm{sym}}(d,u).
\eeq
Moreover,
suppose $X$ and $Q$ are given explicitly,
that $5g+1$ points of degree $1$ on $X$ are given explicitly,
and, for any $d$,
that $n_d$ points of degree $d$ on $X$ are given explicitly.
Suppose also that for each $d,u$ such that $n_{d,u}>0$,
we are given explicitly a symmetric multiplication algorithm
of length $l_{d,u}$ for $\cA_q(d,u)$.
Then, after at most $5g^2$ computations of Riemann-Roch spaces on $X$,
we can construct explicitly a symmetric multiplication algorithm
of length $\sum_{d,u}n_{d,u}l_{d,u}$ for $\cA_q(m,l)$.
\end{enumerate}
\end{theorem}
\begin{proof}
For $1\leq j\leq n_{d,u}$ choose a point $P_{d,u,j}$ of degree $d$
in $X$, such that $P_{d,u,j}\neq P_{d,u',j'}$ if $(u,j)\neq(u',j')$.
This is possible by~\eqref{nd<Nd}
(moreover, in cases b) and c), these $P_{d,u,j}$ are chosen among
the $n_d$ points of degree $d$ given explicitly).
Let then $G=\sum_{d,u}\sum_{1\leq j\leq n_{d,u}}uP_{d,u,j}$,
so that $\deg G=\sum_{d,u}n_{d,u}du$.

\medskip

\emph{Proof of case a).} We suppose $q>5$, and
we can also suppose $g\geq2$, otherwise the conclusion follows
from the results of the previous section.
Let $h=|\Cl^0(X)|$ be the class number of $X$.
Then we also have $h=|\Cl^i(X)|$ for any integer $i$,
where $\Cl^i(X)$ is the set of linear equivalence classes of
divisors of degree $i$ on $X$.
Let also
\beq
\Cl^i_{\textrm{eff}}(X)\subset\Cl^i(X)
\eeq
be the set of linear equivalence classes of effective
divisors of degree $i$ on $X$, or equivalently,
the set of linear equivalence classes of
divisors $D$ of degree $i$ on $X$ such that $l(D)>0$.
We then recall from \cite{NX}, eq.~(6), that
if $A_i$ is the number of effective divisors on $X$, then
\beq
\label{eqNX}
A_{g-1}+2\sum_{i=0}^{g-2}q^{(g-i-1)/2}A_i\leq\frac{h}{(q^{1/2}-1)^2}
\eeq
hence for any $i\leq g-1$
\beq
\label{<h/2}
|\Cl^i_{\textrm{eff}}(X)|\leq A_i\leq\frac{h}{(q^{1/2}-1)^2}<\frac{h}{2}
\eeq
(see also \cite{Ballet2008}, Lemma~2.1, and \cite{BRR2010}, Th.~3.3).
%\begin{equation}
%\label{<h/2}
%\textstyle\textrm{\begin{minipage}{0.85\textwidth}
%\emph{For any $i\leq g-1$, the number of effective divisors
%of degree $i$ on $X$ is at most $\frac{h}{(q^{1/2}-1)^2}$.
%Hence $|\Cl^i_{\textrm{eff}}(X)|\leq\frac{h}{(q^{1/2}-1)^2}<\frac{h}{2}$.}
%\end{minipage}}
%\end{equation}
%In fact \eqref{Ballet2008}, Lemma~2.1, states this only for $i=g-1$,
%but its proof (and more precisely equation~(4) therein)
%also gives the stronger bound $\frac{h}{2q^{(g-1-i)/2}(q^{1/2}-1)^2}$
%for $0\leq i\leq g-2$.
We now let
\beq
t=ml+g-1
\eeq
and we claim that we can find divisors $D_1,D_2$
of degree $t$ such that:
\begin{itemize}
\item[$(i')$] $D_1+D_2-G$ is zero-dimensional
\item[$(ii'_1)$] $D_1-lQ$ is non-special
\item[$(ii'_2)$] $D_2-lQ$ is non-special.
\end{itemize}
Indeed, $(ii'_1)$ means that the linear equivalence class $[D_1-lQ]$
is not in $\Cl^{g-1}_{\textrm{eff}}(X)$, or equivalently,
\beq
\label{D1notintranslate}
[D_1]\not\in\Cl^{g-1}_{\textrm{eff}}(X)+[lQ].
\eeq
But since translation by $[lQ]$ puts $\Cl^{g-1}(X)$ in bijection
with $\Cl^t(X)$, applying \eqref{<h/2} shows the translate
$\Cl^{g-1}_{\textrm{eff}}(X)+[lQ]$ cannot cover all $\Cl^t(X)$,
hence we can find $D_1$ as wished.
Now, this $D_1$ being fixed, $(i')$ and $(ii'_2)$ together mean
\beq
\label{D2notintranslate}
[D_2]\:\not\in\:(\Cl^{2t-\deg G}_{\textrm{eff}}(X)+[G-D_1])\,\cup\,(\Cl^{g-1}_{\textrm{eff}}(X)+[lQ]),
\eeq
where $2t-\deg G\leq g-1$ by~\eqref{G>2n+g-1}.
But again \eqref{<h/2} shows that the union of these translates
has cardinality less than $h/2+h/2$, and we can find $D_2$ as wished.
All this done we can now apply Theorem~\ref{ChCh+gen} and conclude.

\medskip

\emph{Proof of case b).} Suppose we are given a
set $\cS=\{P_0,P_1,\dots,P_{2g}\}$ of $2g+1$ points of degree $1$ on $X$.
As in case a), all we need is to construct divisors $D_1,D_2$ of degree $t$
satisfying $(i')$, $(ii'_1)$, $(ii'_2)$, and apply Theorem~\ref{ChCh+gen}
to conclude.
From \cite{21sep}, Lemma~6, we recall the following:
\begin{equation}
\label{<g}
\textstyle\textrm{\begin{minipage}{0.85\textwidth}
\emph{If $A$ is a divisor on $X$ with $\deg A\leq g-2$ and $l(A)=0$,
there are at most $g$ points $P\in X(\F_q)$ such that $l(A+P)>0$.}
\end{minipage}}
\end{equation}
For $-1\leq i\leq g-1$ we construct a divisor $Y_i$ on $X$ of
degree $ml+i$ such that $l(Y_i-lQ)=0$ iteratively as follows:
\begin{itemize}
\item Start with $Y_{-1}=(ml-1)P_0$, so $l(Y_{-1}-lQ)=0$
for degree reasons.
\item Suppose up to some $i<g-1$ we have found $Y_i$
such that $l(Y_i-lQ)=0$ as wished. Then by~\eqref{<g} there
exists $P\in\cS$ such that $l(Y_i+P-lQ)=0$. We put $Y_{i+1}=Y_i+P$.
\item This ends when $i=g-1$.
\end{itemize}
We can then put $D_1=Y_{g-1}$, so that $(ii'_1)$ is satisfied.

Now for $-1\leq i\leq g-1$ we construct a divisor $Z_i$ on $X$ of
degree $ml+i$ such that $l(Z_i-lQ)=0$ and $l(D_1+Z_i-G)=0$
iteratively as follows:
\begin{itemize}
\item Start with $Z_{-1}=(ml-1)P_0$, so $l(Z_{-1}-lQ)=0$
and $l(D_1+Z_{-1}-G)=0$ for degree reasons
(via hypothesis~\eqref{G>2n+g-1} for the second).
\item Suppose up to some $i<g-1$ we have found $Z_i$
such that $l(Z_i-lQ)=0$ and $l(D_1+Z_i-G)=0$ as wished.
We claim there is a point $P\in\cS$
such that $l(Z_i+P-lQ)=0$ and $l(D_1+Z_i+P-G)=0$.
Indeed by~\eqref{<g} the first can fail at most $g$ times,
and likewise the second can fail at most $g$ times.
We then put $Z_{i+1}=Z_i+P$.
\item This ends when $i=g-1$.
\end{itemize}
We can then put $D_2=Z_{g-1}$, so that $(i')$ and $(ii'_2)$ are satisfied,
and we're done.

\medskip

\emph{Proof of case c).} Suppose we are given a
set $\cT=\{P_0,P_1,\dots,P_{5g}\}$ of $5g+1$ points of degree $1$ on $X$.
From \cite{21sep}, Lemma~9, we recall the following:
\begin{equation}
\label{<4g}
\textstyle\textrm{\begin{minipage}{0.85\textwidth}
\emph{If $A$ is a divisor on $X$ with $\deg A\leq g-3$ and $l(A)=0$,
there are at most $4g$ points $P\in X(\F_q)$ such that $l(A+2P)>0$.}
\end{minipage}}
\end{equation}
Then for $-1\leq i\leq g-1$ we construct a divisor $T_i$ on $X$ of
degree $ml+i$ such that $l(T_i-lQ)=0$ and $l(2T_i-G)=0$
iteratively as follows:
\begin{itemize}
\item Start with $T_{-1}=(ml-1)P_0$, so $l(T_{-1}-lQ)=0$
and $l(2T_{-1}-G)=0$ for degree reasons
(via hypothesis~\eqref{G>2n+g-1} for the second).
\item Suppose up to some $i<g-1$ we have found $T_i$
such that $l(T_i-lQ)=0$ and $l(2T_i-G)=0$ as wished.
We claim there is a point $P\in\cT$
such that $l(T_i+P-lQ)=0$ and $l(2T_i+2P-G)=0$.
Indeed by~\eqref{<g} the first can fail at most $g$ times,
and by~\eqref{<4g} the second can fail at most $4g$ times.
We then put $T_{i+1}=T_i+P$.
\item This ends when $i=g-1$.
\end{itemize}
We can then put $D_1=D_2=T_{g-1}$ and conclude
by Theorem~\ref{ChCh+gen} again.
\end{proof}

\begin{remark}
As explained in the Introduction,
this Theorem~\ref{th=g-1} fixes an error in an article of Ballet.
More precisely, if we take $l=1$,
and we choose all $n_{d,u}$ equal to zero except for $n_{1,1}$,
then case~a) of Theorem~\ref{th=g-1} gives statement (1)
in Th.~2.1 of \cite{Ballet2008} as a special case; and likewise if we
choose all $n_{d,u}$ equal to zero except for $n_{1,1}$ and $n_{2,1}$,
we get its statement (2).

Remark that our proof of case~a) is
structurally the same as Ballet's.
% (and indeed we even used his Lemma~2.1).
The only difference is that we allowed the asymmetry $D_1\neq D_2$,
so $D_1$ and $D_2$ could be constructed one at a time, and in
establishing \eqref{D1notintranslate}
and \eqref{D2notintranslate}
we only had to consider translations  $[D]\mapsto [D]-[A]$
which put $\Cl^*(X)$ in bijection with $\Cl^{*-\deg A}(X)$.
On the other hand Ballet had to consider a map
of the form $[D]\mapsto 2[D]-[G]$ which might be non-injective.
The error in Ballet's \cite{Ballet2008}, Prop.~2.1,
is that he did not
take the possible kernel of this multiplication-by-$2$ map
(that is, the $2$-torsion in the class group)
into account.
As explained in the Introduction,
this error was in fact borrowed from \cite{STV}, and was first
spotted by Cascudo-Cramer-Xing (see \cite{Cascudo}, Chapter~12).

Remark also that case~c) of Theorem~\ref{th=g-1} gives another way
of fixing this error, while keeping symmetry. A drawback is that
the condition $X(\F_q)>5g$ in case~c) imposes serious restrictions
on the curves to be used, hence for some values of $q$, it does
not lead to interesting bounds.

%So, for applications, case~a is often more suitable,
%and indeed it allows us to secure the proof of further results of
%Ballet that were jeopardized by the error in his Th.~2.1:
So, for applications, case~a) is often more suitable,
and indeed it allows us to fix the proof of further bounds of
Ballet that were jeopardized by the error in his Th.~2.1:
\end{remark}

\begin{corollary}
\label{corBallet}
Let $p$ be a prime number and $q=p^r$ a power of $p$, with $q>5$.
Then for all integer $n\geq1$ we have
\beq
\frac{1}{n}\mu_q(n)\leq
\begin{cases}
3\left(1+\frac{2}{p-2}\right) & \textrm{if $r=1$}\\
2\left(1+\frac{2}{\sqrt{q}-2}\right) & \textrm{if $r=2$}\\
3\left(1+\frac{p}{q-2}\right) & \textrm{if $r\geq3$ odd.}
\end{cases}
\eeq
\end{corollary}
\begin{proof}
Use Theorem~\ref{th=g-1}
instead of Th.~2.1 of \cite{Ballet2008},
in the proof of the corresponding cases of Th.~3.1 of \cite{Ballet2008}
and Th.~2.1 and~2.2 of \cite{Balletnote}.

More precisely, Theorem~\ref{th=g-1} with $l=1$, $m=n$, $n_{1,1}=B_1(X/\F_q)$,
and the other $n_{d,u}=0$, replaces Th.~2.1.(1) of \cite{Ballet2008}.
While Theorem~\ref{th=g-1} with $l=1$, $m=n$, $n_{1,1}=B_1(X/\F_q)$,
$n_{2,1}=B_2(X/\F_q)$,
and the other $n_{d,u}=0$, replaces Th.~2.1.(2) of \cite{Ballet2008}.
\end{proof}

\begin{remark}
There is a case of Th.~3.1 of \cite{Ballet2008} that we didn't include
%(on purpose)
in our Corollary. Namely, Th.~3.1 of \cite{Ballet2008} claims that
the bound $\frac{1}{n}\mu_q(n)\leq2\left(1+\frac{2}{\sqrt{q}-2}\right)$
holds for all $r$ even, not only for $r=2$.
The reason for this omission is that there is another error
in the proof of this Th.~3.1 of Ballet,
apart from the oversight of the $2$-torsion
already mentioned.

Indeed in his proof Ballet considers two consecutive prime
numbers $l_1$ and $l_2$ determined by $n$ and he claims that he can apply
his Prop.~3.1.(2) to this $l_2$. However this Prop.~3.1.(2) only states
that \emph{there exists} a prime number $l$ for which its conclusion
holds, not that it holds \emph{for all} prime numbers.
Looking more closely at the proof, we see it works for primes $l$
for which certain points split completely in a certain morphism of curves,
which in turn can be translated as the primes $l$ lying in a certain
arithmetic progression.
However there is no reason that $l_2$ should be in this arithmetic
progression, except in the case $r=2$ where it is trivial.

On the other hand, it is easy to see that this bound, and even a slightly
stronger one, holds at least asymptotically (if not for all $n$),
as will be seen with our fix of the Shparlinski-Tsfasman-Vladut bound
below.
\end{remark}

To end this section, we want to show how the condition $q>5$
in Theorem~\ref{th=g-1}.a) can be relaxed, at the cost of
only weakening condition~\eqref{G>2n+g-1} by a small absolute constant,
independent of $g$.
For this we will use a generalization of~\eqref{<h/2},
that might also be seen as a variant of \cite{BRR2010}, Th.~3.3
and Cor.~3.4.

\begin{lemma}
\label{estim_e}
Let $X$ be a curve of genus $g\geq2$ over $\F_q$,
of class number $h$,
and for any integer $i$ let $A_i$ be the number of effective
divisors of degree $i$ on $X$.
Define an integer $e_q$ as follows:
%\begin{itemize}
%\item If $X(\F_q)=\emptyset$,
%\beq
%e=
%\begin{cases}
%6 & \textrm{if $q=2$}\\
%2 & \textrm{if $q=3$}\\
%1 & \textrm{if $q=4,5$}\\
%0 & \textrm{if $q>5$.}
%\end{cases}
%\eeq
%\item If $X(\F_q)\neq\emptyset$,
\beq
e_q=
\begin{cases}
2 & \textrm{if $q=2$}\\
1 & \textrm{if $q=3,4,5$}\\
0 & \textrm{if $q>5$.}
\end{cases}
%e(2)=3,\quad e(3)=e(4)=e(5)=1,\textrm{ and }e(q)=0\textrm{ for }q>5.
\eeq
%\end{itemize}
Then there is an integer $e$ with $0\leq e\leq e_q$
such that
\beq
\label{e<e_q}
A_{g-e-1}+A_j<h
\eeq
for all $j\leq g+2e-3e_q-1$.
\end{lemma}
\begin{proof}
We first consider the case $q=2$.
If $g=2$, take $e=e_q=2$, so \eqref{e<e_q} is satisfied
since $A_j=0$ for $j<0$. Now suppose $g\geq3$, and
write \eqref{eqNX} in the form
\beq
\label{eqNX2}
A_{g-1}+2\sqrt{2}A_{g-2}+4A_{g-3}+\dots+2(\sqrt{2})^{g-1}A_0\leq(3+2\sqrt{2})h.
\eeq
We proceed by contradiction and suppose that the lemma is false.
This means that the following three inequalities hold:
\beq
\label{Aj}
A_{g-3}+A_{j}\geq h\qquad\text{for some $j\leq g-3$}
\eeq
\beq
\label{Aj'}
A_{g-2}+A_{j'}\geq h\qquad\text{for some $j'\leq g-5$}
\eeq
\beq
\label{Aj''}
\,A_{g-1}+A_{j''}\geq h\qquad\text{for some $j''\leq g-7$.}
\eeq
We multiply \eqref{Aj} by $2$, \eqref{Aj'} by $2\sqrt{2}$, and sum with \eqref{Aj''}, to get:
\beq
\label{anti-eqNX2}
A_{g-1}+2\sqrt{2}A_{g-2}+2A_{g-3}+2A_{j}+2\sqrt{2}A_{j'}+A_{j''}\geq(3+2\sqrt{2})h.
\eeq
Comparing coefficients
(and discussing whether $j=g-3$ or $j\leq g-4$,
and whether $j,j',j''$ are all distinct or some of them are equal)
we see that the left-hand side of \eqref{anti-eqNX2}
is less than or equal to the left-hand side of \eqref{eqNX2}.
To get a contradiction, it suffices to prove that the inequality is strict.

If $g\geq4$, the coefficient of $A_0=1$ in \eqref{anti-eqNX2} is strictly
less than in \eqref{eqNX2}, so the inequality is strict indeed.

Last, if $g=3$, the only way to have equality is to have $j=g-3=0$,
with equality also in \eqref{Aj}, \eqref{Aj'}, and \eqref{Aj''}.
But from this and $A_0=1$ we deduce $h=2=A_1=A_2$.
However $A_1=2$ means there are two points $P_1,P_2$ of degree $1$ on $X$,
and considering the divisors $2P_1,2P_2,P_1+P_2$, we find $A_2\geq3$,
a contradiction.

\smallskip

The case $q=3$ works the same. Write \eqref{eqNX} as
\beq
\label{eqNX3}
A_{g-1}+2\sqrt{3}A_{g-2}+6A_{g-3}+\dots+2(\sqrt{3})^{g-1}A_0\leq(1+\sqrt{3}/2)h<2h.
\eeq
If the lemma were false, one could find $j\leq g-2$ with $A_{g-2}+A_{j}\geq h$,
and $j'\leq g-4$ with $A_{g-1}+A_{j'}\geq h$. Summing these two inequalities
would then contradict \eqref{eqNX3}.

\smallskip

To finish the proof, for $q=4$ or $5$, remark that \eqref{eqNX}
implies $A_i<h/2$ for $i\leq g-2$, so we can take $e=e_q=1$.
And for $q>5$ we find $A_i<h/2$ for $i\leq g-1$, so $e=e_q=0$ works,
as claimed.
\end{proof}

\begin{proposition}
\label{th=g-1+e}
Let $X$ be a curve of genus $g\geq2$ over $\F_q$,
where $q\geq2$ is any prime power,
and let $m,l\geq1$ be two integers.

Suppose that $X$ admits a closed point $Q$ of degree $\deg Q=m$
(a sufficient condition for this is $2g+1\leq q^{(m-1)/2}(q^{1/2}-1)$).

Let $e_q$ be defined as in the previous lemma (remark $e_q\leq2$ in any case).

Consider now a collection of integers $n_{d,u}\geq0$
(for $d,u\geq1$), such that almost all of them are zero, and
that for any $d$,
\beq
n_d=\sum_un_{d,u}\leq B_d(X/\F_q).
\eeq
Then, provided
\beq
\label{G>2n+3e+g-1}
\sum_{d,u}n_{d,u}du\geq2ml+3e_q+g-1,
\eeq
we have
\beq
\mu_q(m,l)\leq\sum_{d,u}n_{d,u}\mu_q(d,u).
\eeq
\end{proposition}
\begin{proof}
We argue essentially as in the proof of Theorem~\ref{th=g-1}.a)
with only a few minor changes.
From the collection of integers $n_{d,u}$ we first construct
a divisor $G$, of degree $\deg G=\sum_{d,u}n_{d,u}du$, as before.
For any integer $i$ we let
\beq
\Cl^i_{\textrm{sp}}(X)\subset\Cl^i(X)
\eeq
be the set of linear equivalence classes of special divisors
on $X$, hence by the Riemann-Roch theorem
$\Cl^i_{\textrm{sp}}(X)=[K_X]-\Cl^{2g-2-i}_{\textrm{eff}}(X)$,
so
\beq
\label{Clsp=Cleff}
|\Cl^i_{\textrm{sp}}(X)|=|\Cl^{2g-2-i}_{\textrm{eff}}(X)|\leq A_{2g-2-i},
\eeq
and by Lemma~\ref{estim_e} there is an $e$ with $0\leq e\leq e_q$ and
\beq
\label{hnew}
|\Cl^{g+e-1}_{\textrm{sp}}(X)|\leq|\Cl^{j}_{\textrm{eff}}(X)|+|\Cl^{g+e-1}_{\textrm{sp}}(X)|\leq A_j+A_{g-e-1}<h
\eeq
for all $j\leq g+2e-3e_q-1$.

Then letting
\beq
t=ml+e+g-1
\eeq
and using \eqref{hnew} instead of \eqref{<h/2}, we can
first find a divisor $D_1$ of degree $t$ such that
\beq
[D_1]\not\in\Cl^{g+e-1}_{\textrm{sp}}(X)+[lQ],
\eeq
ensuring $(ii'_1)$ as in the proof of Theorem~\ref{th=g-1}.a),
and then, a divisor $D_2$ of degree $t$ such that
\beq
[D_2]\:\not\in\:(\Cl^{2t-\deg G}_{\textrm{eff}}(X)+[G-D_1])\,\cup\,(\Cl^{g+e-1}_{\textrm{sp}}(X)+[lQ]),
\eeq
(remark $2t-\deg G\leq g+2e-3e_q-1$ by \eqref{G>2n+3e+g-1}),
ensuring $(i')$ and $(ii'_2)$,
and we conclude as before.
\end{proof}

\begin{remark}
Many results in this part,
concerning ``uniform'' upper bounds,
can still be improved or generalized,
in various directions, for example:
\begin{itemize}
%\item Lemma~\ref{estim_e} can be refined, after a more careful
%case study of the number of points of $X$.
%\item In some cases
%(but of course only for $q\leq5$)
%a closer study of the geometry and arithmetic
%of divisors on $X$ can lead to a better estimate on $e$ in
%\eqref{e<e_q}.
%This then gives (slightly) sharper bounds in Proposition~\ref{th=g-1+e}.
\item Following Ballet's proof, the case $r=2$ in
Corollary~\ref{corBallet} uses modular curves of \emph{prime} genus,
and then relies on Bertrand's postulate (proved by Chebyshev)
for these primes. It is possible to refine both parts of this argument
(allow non-prime values for the genus, and get a finer control on the
gaps between these values),
%also
leading to sharper bounds in this case.
\item Theorem~\ref{th=g-1} (and Proposition~\ref{th=g-1+e})
can also be combined with descent arguments, such as those used
in \cite{BP}, to derive better bounds than the ones in  
Corollary~\ref{corBallet} when $q$ is not a square.
\end{itemize}
All these improvements or generalizations require quite long
technical discussions and are somehow independent of the main
ideas presented in this paper, so they will be treated elsewhere.
\end{remark}

\section{Fixing the Shparlinski-Tsfasman-Vladut asymptotic upper bound}

The Shparlinski-Tsfasman-Vladut upper bound \cite{STV}
concerns the asymptotic
quantities defined below. As explained earlier in the text,
there was a gap in their proof, which our methods allow to fill
(with two independent arguments).

\begin{definition}
If $q$ is a prime power, we let
\beq
\begin{split}
& m_q=\liminf_{n\to\infty}\frac{1}{n}\mu_q(n) \\
& M_q=\limsup_{n\to\infty}\frac{1}{n}\mu_q(n)
\end{split}
\eeq
and their symmetric counterparts $m_q^{\textrm{sym}}$
and $M_q^{\textrm{sym}}$ are defined likewise.
\end{definition}

\begin{definition}
We let $A(q)$ be the largest real number such that there exists a family
of curves $X_s$
over $\F_q$,  of genus $g_s$ going to infinity, with
\beq
\lim_{s\to\infty}\frac{|X_s(\F_q)|}{g_s}=A(q).
\eeq
\end{definition}

\begin{theorem}
\label{STVm}
If $A(q)>1$, then
\beq
m_q\leq2\left(1+\frac{1}{A(q)-1}\right).
\eeq
Moreover, if $A(q)>5$,
then also $m_q^{\textrm{sym}}\leq2\left(1+\frac{1}{A(q)-1}\right)$.
\end{theorem}
\begin{proof}
Consider a family
of curves $X_s$
over $\F_q$,  of genus $g_s$ going to infinity, with
\beq
\label{X/g->A(q)}
\lim_{s\to\infty}\frac{|X_s(\F_q)|}{g_s}=A(q).
\eeq
Given an integer $s$, let
\beq
n(s)=\left\lfloor\frac{1}{2}(|X_s(\F_q)|-g_s-5)\right\rfloor,
\eeq
hence by \eqref{X/g->A(q)}
\beq
\label{n/g=(A-1)/2}
\lim_{s\to\infty}\frac{n(s)}{g_s}=\frac{A(q)-1}{2}.
\eeq
Then for $s$ large enough we have $2g_s+1\leq q^{(n(s)-1)/2}(q^{1/2}-1)$
and we can apply Proposition~\ref{th=g-1+e} with $l=1$
and $m=n(s)$,
and with all $n_{d,u}$ zero except $n_{1,1}=2n(s)+g_s+5$,
to get
\beq
\mu_q(n(s))\leq 2n(s)+g_s+5,
\eeq
which allows to conclude.

If $A(q)>5$, then $|X_s(\F_q)|>5g_s$ for $s$ large enough,
and we can use Theorem~\ref{th=g-1}.c) to conclude likewise.
\end{proof}

\begin{theorem}
\label{STVM}
If $q=p^{2r}\geq 9$ is a square, then
\beq
M_q\leq2\left(1+\frac{1}{\sqrt{q}-2}\right).
\eeq
Moreover, if $q=p^{2r}\geq 49$,
then also $M_q^{\textrm{sym}}\leq2\left(1+\frac{1}{\sqrt{q}-2}\right)$.
\end{theorem}
\begin{proof}
Consider the Shimura curves described in \cite{STV},
pp.~163--166. They form a family of curves $X_s$
over $\F_q$,  of genus $g_s$ going to infinity, with
\beq
\label{X/g->sqrt(q)-1}
\lim_{s\to\infty}\frac{|X_s(\F_q)|}{g_s}=\sqrt{q}-1
\eeq
and
\beq
\label{g'/g->1}
\lim_{s\to\infty}\frac{g_{s+1}}{g_s}=1.
\eeq
Given an integer $n$,
let $s(n)$ be the smallest integer such that
\beq
|X_{s(n)}(\F_q)|\geq 2n+g_{s(n)}-1,
\eeq
hence by~\eqref{X/g->sqrt(q)-1} and \eqref{g'/g->1},
\beq
\label{g=2/(sqrt(q)-2)}
%g_{s(n)}\underset{n\to\infty}{\sim}\frac{2n}{\sqrt{q}-2}.
g_{s(n)}=\frac{2n}{\sqrt{q}-2}+o(n).
\eeq
This then gives
%\eqref{g=2/(sqrt(q)-2} gives
$2g_{s(n)}+1\leq q^{(n-1)/2}(q^{1/2}-1)$ for $n$ large enough,
and we can apply Theorem~\ref{th=g-1}.a) with $l=1$ and $m=n$,
and with all $n_{d,u}$ zero except $n_{1,1}=2n+g_{s(n)}-1$,
to get
\beq
\mu_q(n)\leq 2n+g_{s(n)}-1.
\eeq
This holds for all $n$ large enough, hence dividing by $n$
and using \eqref{g=2/(sqrt(q)-2)} again allows to conclude.

If $q\geq49$, then we can use Theorem~\ref{th=g-1}.c) instead,
and conclude likewise.
\end{proof}

\begin{remark}
As noted in \cite{Balletnote}, we also immediately get from
Corollary~\ref{corBallet}
the bounds $M_p\leq3\left(1+\frac{2}{p-2}\right)$
for $p$ prime,
and $M_q\leq3\left(1+\frac{p}{q-2}\right)$
for $q=p^r$, $r\geq3$ odd.
\end{remark}

\begin{remark}
Prop.~4.1 of \cite{STV} also discusses some constructiveness issues,
which we can improve here.
Suppose that $q\geq9$ is a square, and that
for some increasing sequence of integers $n$, we are given
explicitly a curve $X_n$
of genus
\beq
g_n=\frac{2n}{\sqrt{q}-2}+o(n),
\eeq
together with
a point $Q$ of degree $n$ on $X_n$, and a set $S$ of points
of degree $1$ on $X_n$, such that
\beq
|S|\geq 2n+g_n-1
\eeq
(this is possible, for example, with the curves in \cite{GS}).
Then in the preceding proof we can use Theorem~\ref{th=g-1}.b)
instead of Theorem~\ref{th=g-1}.a), which leads to
a \emph{polynomial time} (in $n$) construction of a multiplication algorithm
for $\F_{q^n}/\F_q$, of length $2n\left(1+\frac{1}{\sqrt{q}-2}\right)+o(n)$
(moreover if $q\geq 49$, we can use Theorem~\ref{th=g-1}.c)
to make the algorithm symmetric).
This is better than Prop.~4.1 of \cite{STV} which, under the same
hypothesis, gives an algorithm
of length $2n\left(1+\frac{4}{\sqrt{q}-5}\right)+o(n)$.
\end{remark}

\begin{remark}
Here we studied the asymptotics of $\mu_q(n)=\mu_q(n,1)$.
We could do the same thing for $\widehat{M}_q(n)=\mu_q(1,n)$,
or more generally for $\mu_q(m,l)$ when both $m$ and $l$ vary.

Note that
the parameters $m$ and $l$ appear at two places in Theorem~\ref{th=g-1}
(or likewise in Proposition~\ref{th=g-1+e}):
\begin{itemize}
\item First, $m$ appears alone when one asks that the curve $X$ should
admit a point $Q$ of degree $m$.
\item Then $m$ and $l$ appear together through the product $ml=\dim\cA_q(m,l)$
in condition~\eqref{G>2n+g-1}.
\end{itemize}

Since the curves in the proofs of Theorems~\ref{STVm} and~\ref{STVM}
all admit at least one point of degree~$1$, we see that the asymptotic
estimates given there for $\mu_q(n)$ also hold for $\widehat{M}_q(n)$:
\beq
\liminf_{n\to\infty}\frac{1}{n}\widehat{M}_q(n)\leq2\left(1+\frac{1}{A(q)-1}\right)\qquad\textrm{for $A(q)>1$}\quad
\eeq
\beq
\limsup_{n\to\infty}\frac{1}{n}\widehat{M}_q(n)\leq2\left(1+\frac{1}{\sqrt{q}-2}\right)\qquad\textrm{for $q\geq9$ a square}
\eeq
(and likewise for their symmetric counterparts).

The same techniques also give asymptotic upper bounds
for
\beq
\frac{1}{ml}\mu_q(m,l).
\eeq
However in order to ensure that the
curves admit a point of degree $m$, we will rely on the sufficient
condition $2g+1\leq q^{(m-1)/2}(q^{1/2}-1)$, and since in the proofs
we will have curves of genus $g$ growing linearly with $n=ml$
(see \eqref{n/g=(A-1)/2} or \eqref{g=2/(sqrt(q)-2)}),
these upper bounds will be valid only in a domain in which $m$
grows at least logarithmically with $ml$.
\end{remark}

\begin{question}
The condition $A(q)>5$ in the last statement of Theorem~\ref{STVm}
(and likewise $q\geq49$ in Theorem~\ref{STVM}) might appear strange.
A natural question is whether the estimate should be valid under the
condition $A(q)>1$ also in the symmetric case.
In fact this condition $A(q)>5$ can be relaxed very slightly,
as shown in \cite{2D-G}. However, to relax it further to $A(q)>1$
would require much deeper results, such as the conjectures proposed
in \cite{ITW2010} on the existence of
curves having many points but few $2$-torsion
in their class group.

This also leads to the following question: do $m_q^{\textrm{sym}}=m_q$,
or $M_q^{\textrm{sym}}=M_q$,
or more generally $\mu_q^{\textrm{sym}}(m,l)=\mu_q(m,l)$
for all $q,m,l$?
Of course this should be put in contrast with the example in
Remark~\ref{ex_mu<musym}.
\end{question}


\begin{thebibliography}{1}

\bibitem{Ballet2008}
S.~Ballet,
\emph{On the tensor rank of the multiplication in the finite fields},
J.~Number Theory~\textbf{128} (2008) 1795--1806.

\bibitem{Balletnote}
S.~Ballet,
``A note on the tensor rank of the multiplication in certain finite fields'',
in:
J.~Chaumine, J.~Hirschfeld \& R.~Rolland (eds.),
\emph{Algebraic geometry and its applications,
Proceedings of the first SAGA conference (Papeete, France, 7-11 May 2007)},
Ser. Number Theory Appl.~\textbf{5}, World Sci. Publ., 2008, 
pp.~332--342.

\bibitem{BL2006}
S.~Ballet \& D.~Le~Brigand,
\emph{On the existence of non-special divisors of degree $g$ and $g-1$ in algebraic function fields over $\F_q$},
J.~Number Theory~\textbf{116} (2006) 293--310.

\bibitem{BLR}
S.~Ballet, D.~Le~Brigand \& R.~Rolland,
``On an application of the definition field descent of a tower of function fields'', in:
F.~Rodier \& S.~Vladut (eds.),
\emph{Proceedings of the Conference "Arithmetic, Geometry and Coding Theory" (AGCT 2005)},
S\'eminaires et Congr\`es~\textbf{21},
Soci\'et\'e Math\'ematique de France, 2010, pp.~187--203.

\bibitem{BP}
S.~Ballet \& J.~Pieltant
\emph{On the tensor rank of multiplication in any extension of $\F_2$},
J.~Complexity~\textbf{27} (2011) 230--245.

\bibitem{BRR2010}
S.~Ballet, C.~Ritzenthaler \& R.~Rolland,
\emph{On the existence of dimension zero divisors in algebraic function fields defined over $\F_q$},
Acta Arith.~\textbf{143} (2010) 377--392. 

\bibitem{BR2004}
S.~Ballet \& R.~Rolland,
\emph{Multiplication algorithm in a finite field and tensor rank of the multiplication},
J. Algebra~\textbf{272} (2004) 173--185.


\bibitem{BourbakiAC9}
N.~Bourbaki,
\emph{\'El\'ements de math\'emathique, Alg\`ebre commutative, Chapitres~8 et~9},
Masson, 1983. Reprint: Springer-Verlag, 2006.

\bibitem{BD}
R.~W.~Brocket \& D.~Dobkin,
\emph{On the optimal evaluation of a set of bilinear forms},
Lin. Alg. Appl.~\textbf{19} (1978) 624--628.

\bibitem{BCS}
P.~B\"urgisser, M.~Clausen \& A.~Shokrollahi,
\emph{Algebraic complexity theory},
Grundlehren der Math. Wissenschaften~\textbf{315},
Springer-Verlag, 1997.

\bibitem{Cascudo}
I.~Cascudo Pueyo,
\emph{On asymptotically good strongly multiplicative linear secret sharing},
Tesis doctoral, Universidad de Oviedo, 2010.

\bibitem{Casc_pers}
I.~Cascudo Pueyo, personal communication.

\bibitem{CCX2011}
I.~Cascudo, R.~Cramer \& C.~Xing,
\emph{The torsion-limit for algebraic functions fields and its application to arithmetic secret sharing},
in: Ph.~Rogaway (ed.)
\emph{Advances in cryptology -- CRYPTO 2011},
Lecture Notes in Comp. Science \textbf{6841},
Springer-Verlag, 2011, pp.~685--705.


\bibitem{CO}
M.~Cenk \& F.~\"Ozbudak,
\emph{On multiplication in finite fields},
J.~Complexity~\textbf{26} (2010) 172--186.

\bibitem{Chaumine}
J.~Chaumine,
``Multiplication in small finite fields using elliptic curves'',
in:
J.~Chaumine, J.~Hirschfeld \& R.~Rolland (eds.),
\emph{Algebraic geometry and its applications,
Proceedings of the first SAGA conference (Papeete, France, 7-11 May 2007)},
Ser. Number Theory Appl.~\textbf{5}, World Sci. Publ., 2008, 
pp.~343--350.

\bibitem{ChenCramer}
H.~Chen \& R.~Cramer,
``Algebraic geometric secret sharing schemes and secure multi-party computations over small fields'',
in: C.~Dwork (ed.),
\emph{Advances in cryptology -- CRYPTO 2006},
Lecture Notes in Comp. Science \textbf{4117},
Springer-Verlag, 2006, pp.~521--536.

\bibitem{ChCh}
D.~V. \& G.~V. Chudnovsky,
\emph{Algebraic complexities and algebraic curves over finite fields},
%Proc. Nat. Acad. Sci. USA \textbf{84} (1987) 1739--1743.
J.~Complexity~\textbf{4} (1988) 285--316.

\bibitem{CDM2000}
R.~Cramer, I.~Damg\r{a}rd \& U.~Maurer,
``Efficient general secure multi-party computation from any linear secret-sharing scheme'',
in: B.~Prenel (ed.)
\emph{Advances in cryptology -- EUROCRYPT 2000},
Lecture Notes in Comp. Science \textbf{1807},
Springer-Verlag, 2000, pp.~316--334.

\bibitem{FZ}
C.~Fiduccia \& Y.~Zalcstein,
\emph{Algebras having linear multiplicative complexities},
J.~Assoc. Comput. Mach.~\textbf{24} (1977) 311--331. 

\bibitem{GS}
A.~Garcia \& H.~Stichtenoth,
\emph{A tower of Artin-Schreier extensions of function fields attaining the Drinfeld-Vladut bound},
Invent. Math.~\textbf{121} (1995) 211--222.

\bibitem{deGr}
H.~F.~de Groote,
\emph{Lectures on the complexity of bilinear problems},
Lecture Notes in Comp. Science~\textbf{245},
Springer-Verlag, 1987.

\bibitem{Hartshorne}
R.~Hartshorne,
\emph{Algebraic geometry},
Graduate Texts in Mathematics~\textbf{52}, Springer-Verlag, 1977.

\bibitem{Karatsuba}
%A.~Karatsuba \& Yu.~Ofman,
%\emph{Multiplication of many-digital numbers by automatic computers},
%Doklady Akad. Nauk SSSR~\textbf{145} (1962) 293--294.
\selectlanguage{russian}
А.~Карацуба \& Ю.~Офман,
%\textsc{А.~Карацуба} \& \textsc{Ю.~Офман},
\emph{Умножение многозначных чисел на автоматах},
%``Умножение многозначных чисел на автоматах'',
%Доклады Академии Наук СССР
Доклады Акад. Наук СССР~\textbf{145} (1962) 293--294.
\selectlanguage{english}
English translation:
A.~Karatsuba \& Yu.~Ofman,
\emph{Multiplication of multi-digit numbers on automata},
Soviet Physics Doklady~\textbf{7} (1963) 595--596.

\bibitem{Landsberg}
J.~M.~Landsberg,
\emph{Geometry and the complexity of matrix multiplication},
Bull. Amer. Math. Soc.~\textbf{45} (2008) 247--284.

\bibitem{LW}
A.~Lempel \& S.~Winograd,
\emph{A new approach to error-correcting codes},
IEEE Trans. Inform. Theory \textbf{23} (1977) 503--508.

\bibitem{NX}
H.~Niederreiter \& C.~Xing,
\emph{Low-discrepancy sequences and global function fields with many rational places},
Finite Fields Appl \textbf{2} (1996) 241--273.

\bibitem{ITW2010}
H.~Randriam,
``Hecke operators with odd determinant and binary frameproof codes beyond the probabilistic bound?'',
in \emph{Proc. of ITW 2010 Dublin -- IEEE Information Theory Workshop, Dublin, Ireland, 2010}.

\bibitem{21sep}
H.~Randriambololona,
\emph{$(2,1)$-separating systems beyond the probabilistic bound},
to appear in Israel J.~Math. ---
\url{http://arxiv.org/abs/1010.5764}

\bibitem{2D-G}
H.~Randriambololona,
\emph{Diviseurs de la forme $2D-G$ sans sections et rang de la multiplication dans les corps finis},
preprint. ---
\url{http://arxiv.org/abs/1103.4335}

\bibitem{Shokro92}
M.~A.~Shokrollahi,
\emph{Optimal algorithms for multiplication in certain finite fields using elliptic curves},
SIAM J.~Comput.~\textbf{21} (1992) 1193--1198.

\bibitem{STV}
I.~Shparlinski, M.~Tsfasman \& S.~Vladut,
``Curves with many points and multiplication in finite fields'',
in:
H.~Stichtenoth \& M.~A.~Tsfasman (eds.),
\emph{Coding theory and algebraic geometry (Luminy, 1991)},
Lecture Notes in Math. \textbf{1518},
Springer-Verlag, 1992, pp.~145--169.

\bibitem{Stichtenoth}
H.~Stichtenoth,
\emph{Algebraic function fields and codes},
Universitext, Springer-Verlag, 1993.

\bibitem{Strassen69}
V.~Strassen,
\emph{Gaussian elimination is not optimal},
Numer. Math.~\textbf{13} (1969) 354--356.

\bibitem{Strassen73}
V.~Strassen,
\emph{Vermeidung von Divisionen},
J.~Reine Angew. Math.~\textbf{264} (1973) 184--202.

\bibitem{Strassen83}
V.~Strassen,
\emph{Rank and optimal computation of generic tensors},
Linear Algebra Appl.~\textbf{52/53} (1983) 645--685.

\bibitem{Waterhouse}
W.~Waterhouse,
\emph{Abelian varieties over finite fields},
Ann. Sci. \'Ecole Norm. Sup.~\textbf{2} (1969) 521--560.

\bibitem{Winograd}
S.~Winograd,
\emph{Some bilinear forms whose multiplicative complexity depends on the field of constants},
Math. Systems Theory~\textbf{10} (1977) 169--180.

\bibitem{Xing2002}
C.~Xing, \emph{Asymptotic bounds on frameproof codes},
IEEE Trans. Inform. Theory \textbf{48} (2002) 2991--2995.

\end{thebibliography}
\end{document}